\RequirePackage{fix-cm}

\documentclass{svjour3}                     
\smartqed  
\usepackage{latexsym}
\usepackage{amssymb}
\usepackage[cmex10]{amsmath}
\usepackage{graphicx}
\usepackage{float}
\usepackage{subfig}
\usepackage[]{placeins}
\usepackage[]{algorithm}
\usepackage{algpseudocode}
\usepackage{multirow}
\usepackage[subnum]{cases}
\usepackage{tabularx}
\usepackage{enumerate}

\usepackage[]{hyperref} 
\usepackage[]{amsrefs} 

\usepackage{todonotes}

\renewcommand{\algorithmicforall}{\textbf{for each}}

\journalname{}
\begin{document}

\title{Fast and Accurate Mining of Correlated Heavy Hitters
}


\author{Italo~Epicoco \and
		Massimo~Cafaro \and
        Marco~Pulimeno
}

\authorrunning{I. Epicoco et al.} 

\institute{I. Epicoco \at University of Salento \\
                     \email{italo.epicoco@unisalento.it} \and
           M. Cafaro \at University of Salento \\
                     \email{massimo.cafaro@unisalento.it} \and        
           M. Pulimeno \at University of Salento \\
                     \email{marco.pulimeno@unisalento.it}
}


\maketitle

\begin{abstract}
The problem of mining Correlated Heavy Hitters (CHH) from a two-dimensional data stream has been introduced recently, and a deterministic algorithm based on the use of the Misra--Gries algorithm has been proposed by Lahiri et al. to solve it. In this paper we present a new counter-based algorithm for tracking CHHs, formally prove its error bounds and correctness and show, through extensive experimental results, that our algorithm outperforms the Misra--Gries based algorithm with regard to accuracy and speed whilst requiring asymptotically much less space. 

\end{abstract}

\section{Introduction}
\label{intro}

Mining heavy hitters (also called frequent items), is a well known and studied data mining task. Algorithms for detecting heavy hitters in a data stream can be classified as being either \emph{counter} or \emph{sketch} based, depending on their main data structure. Counter-based algorithms rely on a fixed number of counters in order to keep track of stream items. 
Sketch-based algorithms, as their name suggests, monitor the input data stream by using a set of counters which are stored in a sketch data structure, tipically a two-dimensional array. In this case, hash functions map items to their corresponding sketch cells. Counter-based algorithms are deterministic and sketch-based ones are randomized, thus providing a probabilistic guarantee. 

Regarding the counter-based algorithms, the first sequential algorithm has been designed by Misra and Gries \cite{Misra82}. About twenty years later, the same algorithm was rediscovered independently by both Demaine et al. \cite{DemaineLM02} (the so-called \emph{Frequent} algorithm) and by Karp et al. \cite{Karp}. Among the recently developed counter-based algorithms we recall here \emph{Sticky Sampling}, \emph{Lossy Counting} \cite{Manku02approximatefrequency} and \emph{Space Saving} \cite{Metwally2006}. Notable examples of sketch-based algorithms are \emph{CountSketch} \cite{Charikar}, \emph{Group Test} \cite{Cormode-grouptest}, \emph{Count-Min} \cite{Cormode05} and \emph{hCount} \cite{Jin03}.

Regarding parallel algorithms, \cite{cafaro-tempesta,Cafaro-Pulimeno-Tempesta} provide message-passing based parallel versions of the Frequent and Space Saving algorithms. Shared-memory algorithms have been designed as well, and we recall here the parallel version of Frequent \cite{Zhang2013}, the parallel version of Lossy Counting \cite{Zhang2012} and the parallel versions of Space Saving \cite{Roy2012,Das2009}. Novel shared-memory parallel algorithms for frequent items were recently proposed in \cite{Tangwongsan2014}. Accelerator based algorithms for frequent items exploiting a GPU (Graphics Processing Unit) include \cite{Govindaraju2005,Erra2012}. 

The problem of mining Correlated Heavy Hitters has been introduced recently \cite{Lahiri2016}  by Lahiri et al. Data mining problems that require to compute correlated heavy hitters may be found in the context of network monitoring and management, as well as anomaly and intrusion detection. As an example, consider the stream of pairs (source address, destination address) of IP packets passing through a router. Identifying the nodes that are responsible for the majority of the traffic through that router (frequent items over a single dimension) could be useful, but it is also interesting to discover, for all of the frequent sources, what are the destinations that receive the majority of connections by the same source. Important sources are detected as frequent items over the first dimension, then we search for the frequent destinations in the context of each one of those sources, i.e., the correlated heavy hitters of the stream. In order to formally state the problem, we recall here preliminary notation and definitions.

\begin{definition}
The frequency $f_{xy}$ of the tuple $(x,y)$ in the stream $\sigma$ is given by $f_{xy} = | \{(i,j) \in \sigma : (x=i) \land (y=j)\} |$.
\end{definition}

\begin{definition}
The frequency $f_x$ of an item which appears as first element in the tuple $(x,y)$ is given by $f_x = | \{(i,j) \in \sigma : (x=i)\} |$. 
\end{definition}

The frequency $f_x$ refers to the frequency of the item $x$ disregarding the second item belonging to the tuple, i.e., the frequency computed when considering the sub-stream induced by the projection of the tuples on the first dimension, which is also referred to as the primary dimension (whose items are referred to as primary items). We are now ready to state the Exact Correlated Heavy Hitters problem.

\begin{problem} Exact Correlated Heavy Hitters problem.
\label{exact-CHH}

In the online setting, given a data stream $\sigma$ of length $N$ made of $(x, y)$ tuples in which the item $x$ is drawn from the universe $\mathcal{U}_1=\{1,...,m_1\}$ and the item $y$ is drawn from the universe $\mathcal{U}_2=\{1,...,m_2\}$, two user-defined thresholds $\phi_1$ and $\phi_2$ such that $0<\phi_1<1$ and $0<\phi_2<1$, the Exact Correlated Heavy Hitters (ECHH) problem requires determining all of the $(x,y)$ tuples such that:

\begin{equation}
f_x > \phi_1 N
\label{frequent-x}
\end{equation}

\noindent and

\begin{equation}
f_{xy} > \phi_2 f_x.
\label{frequent-y-correlated-to-x}
\end{equation}
\end{problem}

The Exact Correlated Heavy Hitters problem can not be solved using limited space and only one pass through the input stream, hence the Approximate Correlated Heavy Hitters problem (ACHH) is introduced \cite{Lahiri2016}. We state the problem as follows.

\begin{problem} Approximate Correlated Heavy Hitters problem.
\label{approx-CHH}

Given a data stream $\sigma$ of length $N$ made of $(x, y)$ tuples in which the item $x$ is drawn from the universe $\mathcal{U}_1=\{1,...,m_1\}$ and the item $y$ is drawn from the universe $\mathcal{U}_2=\{1,...,m_2\}$, two user-defined thresholds $\phi_1$ and $\phi_2$ such that $0<\phi_1<1$ and $0<\phi_2<1$ and two error bounds $\epsilon_1$ and $\epsilon_2$ such that $0 < \epsilon_1 < \phi_1$ and $0 < \epsilon_2 < \phi_2$, the Approximate Correlated Heavy Hitters (ACHH) problem requires determining all of the primary items $x$ such that 

\begin{equation}
f_x > \phi_1 N
\end{equation}

\noindent and no items with 

\begin{equation}
f_x \le (\phi_1 - \epsilon_1)N
\end{equation}

\noindent should be reported; moreover, we are required to determine for each frequent primary candidate $x$, all of the tuples $(x,y)$ such that 

\begin{equation}
f_{xy} > \phi_2 f_x
\end{equation}

\noindent and no tuple $(x,y)$ such that 

\begin{equation}
f_{xy} \le (\phi_2 - \epsilon_2) f_x
\end{equation}

\noindent should be reported. 
\end{problem}

In this paper we present CSSCHH (Cascading Space Saving Correlated Heavy Hitters) a new counter-based algorithm for tracking CHHs in a two-dimensional data stream and solving the ACHH problem, formally prove its error bounds and correctness and show, through extensive experimental results, that our algorithm outperforms the Misra--Gries based algorithm \cite{Lahiri2016} proposed by Lahiri et al. (from now on called MGCHH) with regard to accuracy and speed whilst requiring asymptotically much less space.

The rest of this manuscript is organized as follows. In Section \ref{related_work} we provide an overview of related work, we recall in Section \ref{chh-alg} the MGCHH algorithm introduced in \cite{Lahiri2016} and present our algorithm in Section \ref{alg}. Then, we formally prove its error bound and correctness in Section \ref{correctness}. Next, we analyze our algorithm worst case time and space complexity in Section \ref{analysis}. We compare, from a theoretical perspective, CSSCHH against MGCHH in Section \ref{theory}. Extensive experimental results are reported and discussed in Section \ref{results}. We draw our conclusions in Section \ref{conclusions}.

\section{Related work}
\label{related_work}
The problem of efficiently analyzing two-dimensional data streams in order to gain insights and compute significant statistics has been largely investigated in many forms. 

Gehrke et al. \cite{Gehrke2001}, Ananthakrishna et al.\cite{ananthakrishna2003} and Cormode et al. \cite{cormode2009} refer to the notion of \textit{correlated aggregates} and present solutions tailored to different  contexts. On a two-dimensional stream, i.e., a stream of items' pairs, a correlated aggregate is an aggregate value computed along the second dimension on a set of pairs defined by a particular constraint on the first dimension. A typical correlated aggregate is, for instance, the average value of the items on the second dimension computed for those pairs such that the frequency of the first dimension item is above a fixed threshold. 

Mining heavy hitters can also be applied to streams of pairs, when a pair is regarded as a single item. The problem requires finding all of the items which appear in the stream with a frequency greater than a threshold. \cite{Manerikar2009,Cormode2009freq} present an overview and comparison of the most common frequent items algorithms, while Zhang et al. in \cite{zhang2004} treat the case of multidimensional and hierarchical heavy hitters in the context of network traffic analysis. 

A stream of items' tuples can also be processed in order to find frequent itemsets, i.e., a set of items which appear in the stream with a frequency above a threshold. In the context of streams of pairs we are interested in finding the frequent two-itemsets. Also, strictly correlated is the notion of association rules, which are implications of the type \emph{first item} $\implies$ \emph{second item}; the problem entails searching for pairs of items that are frequent two-itemsets and such that the frequency of the second item  with regard to the number of occurrences of the first item is above a fixed threshold. Several frequent itemsets algorithm in the offline setting have been proposed, we recall here Apriori \cite{agrawal1993mining}, Eclat \cite{zaki2000scalable} and FPGrowth \cite{han2000mining}, while an overview of streaming algorithms that solve the frequent itemsets problem is given in \cite{Cheng2008}.

In \cite{Mirylenka2015}, Mirylenka et al. introduce the notion of \textit{Conditional Heavy Hitters} and compare it with other related problems, such as association rules and Correlated Heavy Hitters, highlighting how solving these problems actually leads to different outputs, each emphasizing particular aspects of the input data stream. A group of algorithms is proposed and experimentally evaluated with respect to the approximate mining of Conditional Heavy Hitters. One of these, FamilyHH, is based on the same approach we use in our algorithm, but the authors conclude that the algorithm is not particularly suitable for that problem. 

Lahiri et al. \cite{Lahiri2016} introduced the notion of Correlated Heavy Hitters and proposed an approximate solution based on the Misra--Gries algorithm. We will refer to the Misra and Gries algorithm as MG, and to the their CHH algorithm as MGCHH. In this paper, we present a new counter-based algorithm for tracking CHHs, with reference to the original problem formulation introduced in \cite{Lahiri2016}. Therefore, we compare our solution against MGCHH. In the next section, we recall the MGCHH algorithm.

\section{The Misra--Gries based CHH algorithm}
\label{chh-alg}

This CHH algorithm, recently introduced in \cite{Lahiri2016} by Lahiri et al., is based on a nested application of the Misra and Gries algorithm \cite{Misra82}. Before delving into the details of MGCHH, we recall first how MG works. Being counter-based, MG keeps track of stream items by using a data structure holding $k$ counters, i.e., pairs (item, frequency). In particular, given a stream $\sigma$ of length $N$, a support threshold $\phi$ to determine frequent items (i.e., those items whose frequency exceeds $\phi N$), and an error threshold $0 < \epsilon < 1$, MG requires at least $k = \frac{1}{\epsilon}$ counters to estimate the frequencies of the items with an error less than or equal to $\epsilon N$.

The MG algorithm works as follows. Upon receiving an item from the stream, if one of the counters is already monitoring the item, then the counter is updated by increasing by one the frequency of the item. If none of the counters is monitoring the item but there is a counter available in the data structure (i.e., a counter which is not monitoring any item), this counter is then assigned the responsibility of monitoring the received item and its corresponding frequency is set to one. Otherwise, all of the counters in the data structure are already in charge of monitoring an item. In this case, since the number of counters can not exceed $k$, all of the counters' frequencies are decremented by one. As a result, the counters whose frequencies after the decrement are reset to zero become available to monitor incoming items (since their items are discarded). It is worth recalling here that MG underestimates the frequency of an item; therefore, a single pass of the algorithm over the stream is not enough for exact identification of frequent items.   

MGCHH is a single-pass algorithm solving Problem~\ref{approx-CHH} as follows. Basically, the algorithm estimates the frequencies of the items occurring in the stream along the primary dimension using a set of counters (primary counters) updated as in the MG algorithm; moreover, another set of counters (secondary counters) is associated to each primary counter to keep track of the frequent items occurring along the secondary dimension and correlated to the primary item. In other words, for each distinct item $d$ along the primary dimension, the algorithm maintains its frequency estimate $\hat{f}_d$ and an embedded MG set of secondary counters related to the sub-stream $\sigma_d$ induced by $d$: $\sigma_d = \{y | (d,y) \in \sigma \}$. The embedded secondary counters are a set of pairs $(s, \hat{f}_{d,s})$, where $s$ is an item occurring in $\sigma_d$ and ${\hat{f}}_{d,s}$ estimates the frequency of the tuple $(d,s)$ in $\sigma$. Alternatively, ${\hat{f}}_{d,s}$ can be seen as the frequency estimate of $s$ in $\sigma_d$. Therefore, the MGCHH actions on ${\hat{f}}_d$ are driven by the item $d$ occurring in $\sigma$ whilst the actions on ${\hat{f}}_{d,s}$ depend instead on the item $s$ occurring in $\sigma_d$.

The MGCHH main data structure is a table $H$, a set of tuples of the form $(d,{\hat{f}}_d, H_d)$, where $d$ is an item along the primary dimension, ${\hat{f}}_d$ is its estimated frequency and $H_d$ is a secondary table which stores the secondary items  occurring with $d$. The $H_d$ table maintains a set of counters (item, frequency) monitoring items occurring along the secondary dimension; for each secondary item $s$, its frequency is determined with regard to $\sigma_d$ and denoted by ${\hat{f}}_{d,s}$.

Let $s_1$ and $s_2$ be respectively the maximum number of primary counters in $H$ and the maximum number of secondary counters in each $H_d$. In MGCHH $s_1$ and $s_2$ depend on the parameters $\phi_1,\phi_2, \epsilon_1,\epsilon_2$, and are set at the beginning of the algorithm. 

When the algorithm starts, its data structures are initialized. The number of counters $s_1$ (for tracking the primary items) and $s_2$ (for tracking correlated items) is selected in order to solve the ACHH problem within the $\epsilon_1$ and $\epsilon_2$ error bounds; moreover, $s_1$ and $s_2$ are also chosen to minimize the total space required, as discussed in \cite{Lahiri2016}. In practice, letting $\alpha = \frac{1+\phi_2}{\phi_1-\epsilon_1}$, if $\epsilon_1 \geq \frac{\epsilon_2}{2 \alpha}$, then MGCHH initializes $s_1 = \frac{2 \alpha}{\epsilon_2}$ counters in order to keep track of the primary frequent items and $s_2 = \frac{2}{\epsilon_2}$ counters to track correlated frequent items; otherwise, if $\epsilon_1 < \frac{\epsilon_2}{2 \alpha}$, then MGCHH sets $s_1 = \frac{1}{\epsilon_1}$ and $s_2 = \frac{1}{\epsilon_2 - \alpha \epsilon_1}$.

 Upon receiving a tuple $(x,y)$ from the stream, the data structures are updated as needed. Depending on the tuple $(x,y)$, the update process works as follows:
 
 \begin{enumerate}
\item
If $x$ is in $H$ (i.e., it is already monitored), and $y$ is in $H_x$ as well, then both ${\hat{f}}_{x}$ and ${\hat{f}}_{x,y}$ are incremented.

\item
If $x$ is in $H$, but $y$ is not in $H_x$ (i.e., it is not monitored) and there is an available counter, then $y$ is added to $H_x$ and its frequency is initialized to one. If no counter is available (i.e., $|H_x| = s_2$), then each counter in $H_x$ is decremented by one. If the frequency of any monitored item goes to zero, the item is evicted from $H_x$. After this operation, the size of $H_x$ is such that $|H_x| \leq s_2$.

\item
If $x$ is not in $H$, and $|H| < s_1$ then a counter is created for $x$ in $H$ setting its frequency ${\hat{f}}_{x}$ to one, and initializing $H_x$ with the counter $(y,1)$. If $|H| = s_1$, then for each monitored item $d \in H$, its frequency $\hat f_d$ is decremented by one; if this decrement causes ${\hat{f}}_d$ to become zero, then the counter monitoring $d$ is discarded and removed from $H$. Otherwise, an arbitrary item $s$ is randomly selected from $H_d$ such that ${\hat{f}}_{d,s} > 0$ and ${\hat{f}}_{d,s}$ is decremented by one. If ${\hat{f}}_{d,s}$ goes to zero, the item $s$ is discarded from $H_d$. This further decrement guarantees that the sum of the $\hat f_{d,s}$ frequencies within $H_d$ is less than or equal to $\hat f_d$.
\end{enumerate}

Finally, in order to report the CHHs in the stream, a query can be posed to the data structures as follows. For each primary item $d \in H$, if ${\hat{f}}_d \ge ({\phi}_1- \frac{1}{s_1})N$ then the algorithm searches for the secondary items $s \in H_d$ such that ${\hat{f}}_{d,s} \ge ({\phi}_2- \frac{1}{s_2}){\hat{f}}_d - \frac{N}{s_1}$ and returns the corresponding $(d,s)$ tuples.

\section{A Space Saving based algorithm}
\label{alg}

Our Cascading Space Saving Correlated Heavy Hitters (CSSCHH) algorithm exploits the basic ideas of the Space Saving algorithm \cite{Metwally2006}, combining two Space Saving stream summaries for tracking the primary item frequencies and the tuple frequencies. We refer to our algorithm as cascading Space Saving since it is based on the use of two distinct and independent applications of Space Saving. 

It is worth noting here immediately that it is not possible to follow the same approach used in MGCHH, i.e., a nested application of Space Saving in place of Misra--Gries. The reason is that upon arrival of a primary item which is not monitored, if the primary summary is already full then Space Saving must evict the item monitored by the counter with the minimum frequency and replace it with the newly arrived item, incrementing the counter. However, in the nested approach, substituting the item also requires updating the embedded secondary summary but, in the Space Saving case, we can not reuse the secondary summary because it refers to the evicted item, not to the newly arrived one.

Therefore, we use two independent Space Saving stream summaries as data structures. The first, denoted by $\mathcal{S}^p$, and referred to as the primary stream summary, monitors a subset of primary items which appears in the stream through the use of $k_1$ distinct counters. The second, denoted by $\mathcal{S}^t$, includes $k_2$ counters and monitors a subset of the tuples which appear in the stream. 

The counters are updated in order to accurately estimate the items' frequencies and a lightweight data structure is exploited to keep the elements sorted by their estimated frequencies. A detailed description of the Space Saving algorithm is given in \cite{Metwally2006}. Here, we briefly recall how Space Saving works and its main properties. 

A stream summary $\mathcal{S}$ is a data structure used to monitor $k$ distinct items and includes $k$ counters. We denote with $c_j$ the $j$th counter and by $c_j.i$ and $c_j.f$ respectively the item monitored by the $j$th counter and its corresponding estimated frequency. When processing an item which is already monitored by a counter, its estimated frequency is incremented by one. When processing an item which is not monitored, there are two possibilities. If a counter is available, it will be in charge of monitoring the item and its estimated frequency is set to one. Otherwise, if all of the counters are already occupied (their frequencies are different from zero), the counter storing the item with minimum frequency is incremented by one. Then the monitored item is replaced by the new item. This is because since an item which is not monitored can not have a frequency greater than the minimal frequency. The complexity of the Space Saving update procedure is $O(1)$ in the worst case, as proved by its authors. 


Let $N$ be the length of the input stream, $\sum_{c_i \in \mathcal{S}} c_i.f$ the sum of the counters in $\mathcal{S}$, $k=|\mathcal{S}|$ the number of counters in $\mathcal{S}$, $f_v$ the exact frequency of an item $v$, $\hat f_v$ its estimated frequency and $\hat f^{min}$ the minimum frequency in $\mathcal{S}$. Then, the following relations hold for Space Saving:

\begin{equation}
	\begin{aligned}
	&\sum_{c_i \in \mathcal{S}} c_i.f = N, \\
	&\hat f_v - f_v \leq \hat f^{min} \leq \frac{N}{k}.
	\end{aligned}
\end{equation}

Our CSSCHH algorithm starts by initializing the $\mathcal{S}^p$ primary stream summary data structure allocating $k_1$ counters and the correlated $\mathcal{S}^t$ stream summary allocating $k_2$ counters. We shall explain in Section \ref{correctness} how exactly the values of $k_1$ and $k_2$ are derived. Algorithm \ref{CSSCHH Init} presents the pseudocode related to the initialization phase of CSSCHH.

\begin{algorithm}
\begin{algorithmic}[1]
\Require Threshold for primary itmes ${\phi}_1$; threshold for correlated items ${\phi}_2$; tolerance for primary items ${\epsilon}_1$;
      tolerance for correlated items ${\epsilon}_2$.
\Ensure Properly initialized $\mathcal{S}^p$ and $\mathcal{S}^t$ stream summaries

\Procedure {CSSCHH-Init}{${\phi}_1,{\phi}_2,{\epsilon}_1,{\epsilon}_2$}

\State $\beta \leftarrow \frac{1}{\epsilon_2 \phi_1}$ 
\State $\gamma \leftarrow \frac{\epsilon_2 + \phi_2}{\epsilon_2 \phi_1}$
\State $k_1 \leftarrow \max\left\{\frac{1}{\epsilon_1}, \gamma + \sqrt{\beta \gamma}\right\}$
\State $k_2 \leftarrow \beta\frac{k_1}{k_1 - \gamma}$
\State Allocate $k_1$ counters for $\mathcal{S}^p$
\State Allocate $k_2$ counters for $\mathcal{S}^t$
\State \Return $\mathcal{S}^p$ and $\mathcal{S}^t$
\EndProcedure
\caption{CSSCHH Init}
\label{CSSCHH Init}
\end{algorithmic}
\end{algorithm}

In CSSCHH algorithm, a tuple $(x,y)$ is processed by updating two stream summaries, $\mathcal{S}^p$ and $\mathcal{S}^t$. The primary item $x$ is used to update the primary stream summary $\mathcal{S}^p$; the tuple $(x,y)$ is also considered as a single item and it is used to update the correlated stream summary $\mathcal{S}^t$. Since for each tuple in the stream both stream summaries are updated by means of the Space Saving update procedure, we inherit its properties. Let $f_x$ and $\hat f_x$ denote the exact and estimated frequency of the primary item $x$, and let $f_{xy}$ and $\hat f_{xy}$ denote the exact and estimated frequency of the tuple $(x,y)$; moreover, denoting by $c^p_i$ and $c^t_i$ the $i$th counter in the primary and in the correlated stream summary, and denoting by $\hat {f^p}^{min}$ and $\hat {f^t}^{min}$ the minimum frequency in the primary and correlated stream summary, the following relations hold:

\begin{equation}
	\begin{aligned}
	&\sum_{c^p_i \in \mathcal{S}^p} c^p_i.f = N, \\
	&\hat f_x - f_x \leq \hat {f^p}^{min} \leq \frac{N}{k_1}, \\
	\end{aligned}
\end{equation}

\begin{equation}
	\begin{aligned}
	&\sum_{c^t_i \in \mathcal{S}^t} c^t_i.f = N, \\
	&\hat f_{xy} - f_{xy} \leq \hat {f^t}^{min} \leq \frac{N}{k_2}.
	\end{aligned}
\end{equation}

 The update procedure of CSSCHH is presented in Algorithm~\ref{CSSCHH Update}. 

\begin{algorithm}
\begin{algorithmic}[1]
\Require $x, y$, the items of a tuple.
\Ensure Update of $\mathcal{S}^p$ and $\mathcal{S}^t$ stream summaries.
\Procedure {CSSCHH-Update}{$\mathcal{S}^p, \mathcal{S}^t, x, y$}
	\State \Call{SpaceSavingUpdate}{$\mathcal{S}^p$, $x$}
	\State \Call{SpaceSavingUpdate}{$\mathcal{S}^t, (x, y)$}
\EndProcedure
\caption{CSSCHH Update}
\label{CSSCHH Update}
\end{algorithmic}
\end{algorithm}

In order to retrieve the correlated heavy hitters, a query is posed to both stream summaries. The query procedure internally uses two lists, $F$ and $C$. The former stores primary items and their estimated frequencies $(r, \hat f_r)$. The latter stores CHHs $(r, s, \hat f_{rs})$ in which $r$ is a primary frequent item, $s$ the correlated frequent item candidate and $\hat f_{rs}$ the estimated frequency of the tuple $(r, s)$.

The query algorithm inspects all of the $k_1$ counters in the $\mathcal{S}^p$ stream summary. If the frequency of the monitored item is greater than the selection criterion (i.e., $c^p_j.f > \phi_1 N$), then we add the monitored item $r = c^p_j.i$ and its estimated frequency $\hat f_r = c^p_j.f$ to $F$. 

The algorithm inspects now all of the $k_2$ counters of the $\mathcal{S}^t$ stream summary. The monitored items in $\mathcal{S}^t$ are the tuples $(r,s)$. We check if the primary item $r$ is a primary frequent item candidate (i.e., if $r \in F$); if this condition is true and the tuple estimated frequency is greater than the selection criterion (i.e., $c^t_j.f > \phi_2 (\hat f_r - \frac{N}{k_1})$), then the triplet $(r, s, \hat f_{rs})$ is added to $C$. The Query procedure is presented as Algorithm~\ref{CSSCHH Query}. 

\begin{algorithm}
\begin{algorithmic}[1]
\Require $\mathcal{S}^p$ and $\mathcal{S}^t$ stream summaries.
\Ensure Set of correlated frequent items $C$ 
\Procedure {CSSCHH-Query}{$\mathcal{S}^p$, $\mathcal{S}^t$}
\State $F \leftarrow \emptyset$
\ForAll{$c^p_j \in \mathcal{S}^p$}
	\State $r \leftarrow c^p_j.i$;  $\hat f_r \leftarrow c^p_j.f$
	\If{$\hat f_r  > \phi_1 N$}	
		\State $F \leftarrow F \cup \{(r, \hat f_r)\}$
	\EndIf
\EndFor
\ForAll{$c^t_j \in \mathcal{S}^t$}
	\State $(r,s) \leftarrow c^t_j.i$; $\hat f_{rs} \leftarrow c^p_j.f$
	\If{$r \in F \land (\hat f_{rs}  > \phi_2 (\hat f_r - \frac{N}{k_1}))$}
		\State $C \leftarrow C \cup \{(r, s, \hat f_{rs})\}$
	\EndIf
\EndFor
\State \Return $C$
\EndProcedure
\caption{CSSCHH Query}
\label{CSSCHH Query}
\end{algorithmic}
\end{algorithm}

\section {Correctness}
\label{correctness}
We are going to formally prove the correctness of our algorithm. The main results of this section are the following two theorems. 

\begin{theorem}
\label{thm_primfreq}
The CSSCHH algorithm reports all of the primary items $x$ whose exact frequency $f_x$ is greater than the threshold, i.e., $f_x > \phi_1 N$ and no items whose exact frequency is such that $f_x \le (\phi_1 - \frac{1}{k_1})N$.
\end{theorem}
 
\begin{proof}
The algorithm determines all of the primary frequent candidates through the selection criterion $\hat f_x > \phi_1 N$. Since the stream summary provides an overestimation of the frequency $\hat f_x \geq f_x$, if the exact frequency of an item is greater than the threshold, its estimated frequency will be greater as well: $\hat f_x \geq f_x > \phi_1 N$, hence the item will be selected and this proves the first part of the theorem. 

The second part of the theorem states that, given an item $x$, if its exact frequency is $f_x \le (\phi_1 - \frac{1}{k_1})N$, which can be  rewritten as 
\begin{equation}
\label{eqn_falsepositive}
f_x + \frac{1}{k_1}N \le \phi_1N,
\end{equation}
\noindent then the item will not be selected; hence we must prove that its estimated frequency is less than $\hat f_x \leq \phi_1 N$. By using the Space Saving properties we know that the estimate error provided by the stream summary $\mathcal{S}^p$ is bounded by $\frac{N}{k_1}$:

\begin{equation}
	\begin{aligned}
	\hat f_x - f_x \leq \frac{1}{k_1} N \Rightarrow \hat f_x \leq f_x + \frac{1}{k_1} N,\\
	\end{aligned}
\end{equation}

\noindent hence, by using Eq.~(\ref{eqn_falsepositive}) we have 
\begin{equation}
	\begin{aligned}
	\hat f_x \leq f_x + \frac{1}{k_1} N \le \phi_1 N,
	\end{aligned}
\end{equation}

\noindent which proves the theorem. Moreover, this theorem also implies that for all of the primary frequent candidates it holds that:

\begin{equation}
	\begin{aligned}
	f_x > (\phi_1 - \frac{1}{k_1})N.
	\end{aligned}
\end{equation}
 
\end{proof} \qed

\begin{theorem}
\label{thm_corrfreq}
All of the tuples $(x,y)$ with the item $x$ reported as primary frequent candidate and with exact frequency $f_{xy}$ greater than the threshold ($f_{xy} > \phi_2 f_x$) are reported as correlated heavy hitter candidate. No tuple with a primary item $x$ reported as frequent primary candidate and with exact frequency less than $f_{xy} \le (\phi_2 - \frac{k_2 \phi_2 + k_1}{k_2(k_1 \phi_1 -1)})f_x$ is reported as correlated heavy hitter candidate.
\end{theorem}
 
\begin{proof}
The algorithm determines a correlated heavy hitter candidate $(x,y)$ only if the primary item $x$ as been reported as primary frequent item candidate and if its estimated frequency is greater than the selection criterion $\hat f_{xy} > \phi_2 (\hat f_x - \frac{N}{k_1})$. We must prove that those tuples whose exact frequency is greater than the threshold $f_{xy} > \phi_2 f_x$ are reported by the algorithm and hence their estimated frequency is greater than the selection criterion $\hat f_{xy} > \phi_2 (\hat f_x - \frac{N}{k_1})$. If $f_{xy} > \phi_2 f_x$ is true, then $\hat f_{xy} > \phi_2 f_x$ is also true since the stream summary $\mathcal{S}^t$ provides an overestimation of the tuple frequency. Now, since $x$ is reported as primary frequent candidate, from Theorem~\ref{thm_primfreq} we have that $f_x \geq \hat f_x - \frac{N}{k_1}$ and it holds that 

\begin{equation}
	\hat f_{xy} > \phi_2 f_x \geq \phi_2 (\hat f_x - \frac{N}{k_1}).
\end{equation}

Since the frequency estimate is greater than the selection criterion, the tuple will be reported and this proves the first part of the theorem.

The second part of the theorem states that those items with an exact frequency such that $f_{xy} \le (\phi_2 - \frac{k_2 \phi_2 + k_1}{k_2(k_1 \phi_1 -1)})f_x$ will not be reported, hence we must prove that their estimate frequency is less than or equal to the selection criterion i.e., $\hat f_{xy} \leq \phi_2 (\hat f_x - \frac{N}{k_1})$.
Since the algorithm first filters the tuples retaining only the ones whose primary item belongs to the primary frequent item candidates, from Theorem~\ref{thm_primfreq} it follows that for all of the primary frequent candidates it holds that:

\begin{equation}
\label {fx_bound}
	\begin{aligned}
	f_x > (\phi_1 - \frac{1}{k_1}) N \Rightarrow f_x > \frac{\phi_1 k_1 - 1}{k_1}N.
	\end{aligned}
\end{equation}

To prove the theorem we start assuming the exact frequency is such that 

\begin{equation}
	\begin{aligned}
	f_{xy} \le (\phi_2 - \frac{k_2 \phi_2 + k_1}{k_2(k_1 \phi_1 -1)})f_x.
	\end{aligned}
\end{equation}

Due to the $\mathcal{S}^t$ stream summary properties, the error of the tuple frequency estimate is bounded by $\frac{N}{k_2}$, so that $\hat f_{xy} - \frac{N}{k_2} \leq f_{xy}$ and it holds that:

\begin{equation}
	\begin{aligned}
	\hat f_{xy} - \frac{N}{k_2} \leq f_{xy} \le \phi_2 f_x - \frac{k_2 \phi_2 + k_1}{k_2(k_1 \phi_1 -1)} f_x.
	\end{aligned}
\end{equation}

Since the frequency estimate is an overestimation (i.e., $f_x \leq \hat f_x$), 

\begin{equation}
	\begin{aligned}
	\hat f_{xy} - \frac{N}{k_2} \le \phi_2 \hat f_x - \frac{k_2 \phi_2 + k_1}{k_2(k_1 \phi_1 -1)} f_x;
	\end{aligned}
\end{equation}

\noindent using Eq.~(\ref{fx_bound}) we can write:

\begin{equation}
	\begin{aligned}
	\hat f_{xy} - \frac{N}{k_2} &\le \phi_2 \hat f_x - \frac{k_2 \phi_2 + k_1}{k_2(k_1 \phi_1 -1)} \frac{\phi_1 k_1 - 1}{k_1} N, \\
	\hat f_{xy} - \frac{N}{k_2} &\le \phi_2 \hat f_x - \frac{k_2 \phi_2 + k_1}{k_2 k_1} N, \\	
	\hat f_{xy} - \frac{N}{k_2} &\le \phi_2 \hat f_x - \frac{\phi_2}{k_1}N - \frac{N}{k_2}, \\	
	\hat f_{xy} &\le \phi_2 (\hat f_x - \frac{N}{k_1}).
	\end{aligned}
\end{equation}

Taking into account that the estimated frequency is less than the selection criterion, the corresponding tuple will not be reported, proving the theorem. Moreover, this theorem also implies that for all of the correlated heavy hitter candidates it holds that:

\begin{equation}
f_{xy} > (\phi_2 - \frac{k_2 \phi_2 + k_1}{k_2(k_1 \phi_1 -1)})f_x.
\end{equation}
 
\end{proof} \qed

Theorems~\ref{thm_primfreq} and \ref{thm_corrfreq} can be used for tuning the stream summary sizes $k_1$ and $k_2$. The ACHH problem poses a constraint about the tolerance $\epsilon_1$ on the number of primary frequent false positives and a corresponding constraint about the tolerance $\epsilon_2$ on the number of correlated heavy hitter false positives. Therefore, $k_1$ and $k_2$ are also subject to the following constraints:

\begin {enumerate}[C1]
\item The ACHH problem does not admit false negatives, hence all of the real primary frequent items must be reported. The maximum number of primary frequent items is $\frac{1}{\phi_1}$ hence
\begin{equation}
	k_1 \geq \frac{1}{\phi_1}.
\end{equation}

\item The ACHH problem allows primary false positives only with a tolerance given by $\epsilon_1N$, hence for all of the primary frequent candidates it must be $f_x \geq (\phi_1 -\epsilon_1)N$. Using Theorem~\ref{thm_primfreq} we need to impose:

\begin{equation}
	\frac{1}{k_1} \leq \epsilon_1 \Rightarrow k_1 \geq \frac{1}{\epsilon_1}.
\end{equation}

\item The ACHH problem does not admit false negatives, hence all of the real correlated heavy hitters must be reported. The maximum number of correlated heavy hitters is $\frac{1}{\phi_1\phi_2}$,  hence
\begin{equation}
	k_2 \geq \frac{1}{\phi_1\phi_2}.
\end{equation}

\item The ACHH problem allows correlated false positives only with a tolerance given by $\epsilon_2 f_x$, hence for all of the correlated heavy hitter candidates it must be $f_{xy} \geq (\phi_2 -\epsilon_2)f_x$. Using Theorem~\ref{thm_corrfreq} we need to impose:

\begin{equation}
\label{k2_constraint}
	\frac{k_2 \phi_2 + k_1}{k_2(k_1 \phi_1 -1)} \leq \epsilon_2.
\end{equation}
 
Solving Eq.~(\ref{k2_constraint}) w.r.t. $k_2$, we have

\begin{equation}
\label{k2_constraint_b}
	k_2(\epsilon_2 k_1 \phi_1 - \epsilon_2 -\phi_2) \geq k_1
\end{equation}

Since $k_2$ and $k_1$ represent the sizes of the stream summaries, both are positives integers. Therefore, a solution to Eq.~(\ref{k2_constraint_b}) is feasible only when the left hand side term is positive. To summarize, the current constraint is expressed by the following equations:

\begin{equation}
	\begin {aligned}
	\epsilon_2 k_1 \phi_1 - \epsilon_2 -\phi_2 > 0 \Rightarrow &k_1 > \frac{\epsilon_2 + \phi_2}{\epsilon_2 \phi_1},\\
	&k_2 \geq \frac{k_1}{\epsilon_2 k_1 \phi_1 - \epsilon_2 -\phi_2}.
	\end {aligned}
\end{equation}

\end {enumerate}

Constraint C1 can be ignored since it is already embodied by constraint C2; indeed, the problem requires that $\epsilon_1 < \phi_1$. Constraint C3 can be ignored as well since it can be easily proved that $\frac{k_1}{\epsilon_2 k_1 \phi_1 - \epsilon_2 -\phi_2} > \frac{1}{\phi_1\phi_2}$ for any value of the input parameters.

Introducing the terms $\beta = \frac{1}{\epsilon_2 \phi_1}$ and $\gamma = \frac{\epsilon_2 + \phi_2}{\epsilon_2 \phi_1}$, we have to determine $k_1$ and $k_2$ such that their sum is minimized subject to the following constraints:

\begin{equation}
	\begin {aligned}
	&k_1 \geq \frac{1}{\epsilon_1}, \\
	&k_1 > \gamma, \\
	&k_2 \geq \beta\frac{k_1}{k_1 - \gamma}.
	\end {aligned}
\end{equation}

The optimal values for $k_1$ and $k_2$ are obtained by solving a constrained minimization problem. Introducing a variable substitution $r = k_1+k_2$ and $s=k_1$ the constrained minimization problem is formulated as follows:

\begin{equation}
\begin{aligned}
& {\text{minimize}}
&&  r \\
& \text{subject to}
&& s \geq \frac{1}{\epsilon_1}, \\
&&& s > \gamma, \\
&&& r \geq s + \beta\frac{s}{s - \gamma}.
\end{aligned}
\end{equation}

From the last constraint we can deduce that the minimum value of $r$ must belong to the curve $r = s + \beta\frac{s}{s - \gamma}$ which attains its minimum for $s = \gamma + \sqrt{\beta \gamma}$. Taking into account both constraints on $s$, the minimum is reached when

\begin{equation}
	\begin {aligned}
	s &= \max\left\{\frac{1}{\epsilon_1}, \gamma + \sqrt{\beta \gamma}\right\}, \\
	r &= s + \beta\frac{s}{s - \gamma}.
	\end {aligned}
\end{equation}

Therefore, the corresponding $k_1$ and $k_2$ values are
\begin{equation}
	\begin {aligned}
	k_1 &= \max\left\{\frac{1}{\epsilon_1}, \gamma + \sqrt{\beta \gamma}\right\}, \\
	k_2 &= \beta\frac{k_1}{k_1 - \gamma}.
	\end {aligned}
	\label{eq_k1k2}
\end{equation}

These are the values set by Algorithm \ref{CSSCHH Init} in order to initialize the CSSCHH stream summaries data structure by using the minimum number of counters, and, consequently, of space required to solve the ACHH problem.

\section{Space and Time complexity}
\label{analysis}

In this section, we analyze the worst case time and space complexity of our algorithm. Regarding the initialization phase (Algorithm \ref{CSSCHH Init}), the worst case complexity is  clearly $O(1)$ since initialization consists of just a few assignments, each one requiring at most $O(1)$ time. 

The update procedure (Algorithm \ref{CSSCHH Update}) requires constant time as well. Indeed, each one of the two calls to {\sf SPACESAVINGUPDATE} requires at most $O(1)$.

Finally, a query (Algorithm \ref{CSSCHH Query}) requires time at most $O(k_1 + k_2)$. Indeed, the first part of the query is just a linear scan of the $k_1$ counters related to the primary stream summary $\mathcal{S}^p$, in which we check, for each counter, if the corresponding monitored item's frequency exceed the selection criterion. When the check succeeds, the item and its estimated frequency are added to the hash table $F$. Since checking the condition can be done in constant time, this part of the query requires time at most $O(k_1)$. 

Next, we inspect the correlated stream summary $\mathcal{S}^t$. Again, this is just a linear scan. For each counter, we retrieve the stored tuple $(r, s)$, and its estimated frequency. Then, we search for the primary item $r$ in $F$; if the item belongs to $F$ and if the condition required for the tuple $(r, s)$ to be considered a CHH is verified, then we update the list $C$ holding the CHHs that shall be returned to the user. Searching in $F$ requires constant time (since $F$ is implemented as an hash table), and verifying the CHH condition requires constant time as well, the second part of the query requires time at most $O(k_2)$. It follows that, overall, the query requires in the worst case time at most $O(k_1 + k_2)$.

Regarding the space complexity, it is clear that the total space required is at most $O(k_1 + k_2)$, since we use $k_1$ counters for $\mathcal{S}^p$ and $k_2$ counters for $\mathcal{S}^t$. More specifically, in order to express the space complexity with regard to the input parameters, we distinguish two cases as in Eq.~(\ref{eq_k1k2}). When $1/\epsilon_1 \leq \gamma + \sqrt{\beta \gamma}$, we have

\begin{equation}
	k_1 + k_2 = \frac{2 \sqrt{\epsilon_2+ \phi_2}+ \epsilon_2+\phi_2+1}{\epsilon_2 \phi_1} = O\left(\frac{1}{\epsilon_2 \phi_1}\right);
\end{equation}

\noindent otherwise, for $1/\epsilon_1 > \gamma + \sqrt{\beta \gamma}$, it holds that

\begin{equation}
	k_1 + k_2 = \frac{1}{\epsilon_1}+\frac{1}{\epsilon_1 \epsilon_2 \phi_1 \left(\frac{1}{\epsilon_1} - \frac{\epsilon_2+\phi_2}{\epsilon_2 \phi_1}\right)} < \frac{1}{\epsilon_1}+\frac{1}{\epsilon_1 \sqrt{\epsilon_2+\phi_2}} = O\left(\frac{1}{\epsilon_1 \sqrt{\epsilon_2}}\right).
\end{equation}

\section{Theoretical Comparison}
\label{theory}

In this section we compare MGCHH and our algorithm CSSCHH from a theoretical perspective, before presenting the results of the experiments that we have carried out. We begin by comparing the space complexity and how many counters are required by both algorithms to guarantee their error bounds.

Let $\alpha = \frac{1+\phi_2}{\phi_1-\epsilon_1}$, then, if $\epsilon_1 \geq \frac{\epsilon_2}{2 \alpha}$, MGCHH requires $s_1 = \frac{2 \alpha}{\epsilon_2}$ counters in order to keep track of the primary frequent items, and $s_2 = \frac{2}{\epsilon_2}$ counters to track correlated frequent items; otherwise (if $\epsilon_1 < \frac{\epsilon_2}{2 \alpha}$), $s_1 = \frac{1}{\epsilon_1}$ and $s_2 = \frac{1}{\epsilon_2 - \alpha \epsilon_1}$. 

In the former case the space complexity of MGCHH is $O(\frac{1}{(\phi_1 - \epsilon_1) \epsilon_2^2})$, and in the latter case its space complexity is $O(\frac{1}{\epsilon_1 \epsilon_2})$ \cite{Lahiri2016}.

MGCHH requires a total of $s_1 + (s_1 s_2)$ counters, since the algorithm consists of a nested application of the Frequent algorithm: besides the $s_1$ counters for primary frequent items, there are $s_1 s_2$ counters for correlated items. Indeed, for each primary counter there is an entire summary consisting of $s_2$ counters. 

Our CSSCHH algorithm requires $k_1 = \max\left\{\frac{1}{\epsilon_1}, \gamma + \sqrt{\beta \gamma}\right\}$ counters for the primary frequent items and $k_2 = \beta\frac{k_1}{k_1 - \gamma}$ for the correlated frequent items, where $\beta = \frac{1}{\epsilon_2 \phi_1}$ and $\gamma = \frac{\epsilon_2 + \phi_2}{\epsilon_2 \phi_1}$. 

As shown in the previous section, the space complexity of our algorithm is $ O\left(\frac{1}{\epsilon_2 \phi_1}\right)$ when $1/\epsilon_1 \leq \gamma + \sqrt{\beta \gamma}$, and it is $O\left(\frac{1}{\epsilon_1 \sqrt{\epsilon_2}}\right)$ when $1/\epsilon_1 > \gamma + \sqrt{\beta \gamma}$. It is immediate verifying that our algorithm requires asymptotically less space than MGCHH. From a practical perspective, it's worth noting here that MGCHH is also subject to the constraint $\epsilon_1 < \phi_1 / 2$, whilst our algorithm is not. 

Figure \ref{counters} depicts the number of counters required by MGCHH and CSSCHH (using a logarithmic scale on the $z$ axis) fixing $\phi_1 = 0.01$, $\phi_2 = 0.01$ and letting $\epsilon_1$ varying up to $\phi_1 / 2$,  $\epsilon_2$ varying up to $\phi_2$. 

\begin{figure}[hbt]
\centering
  \includegraphics[scale=0.6]{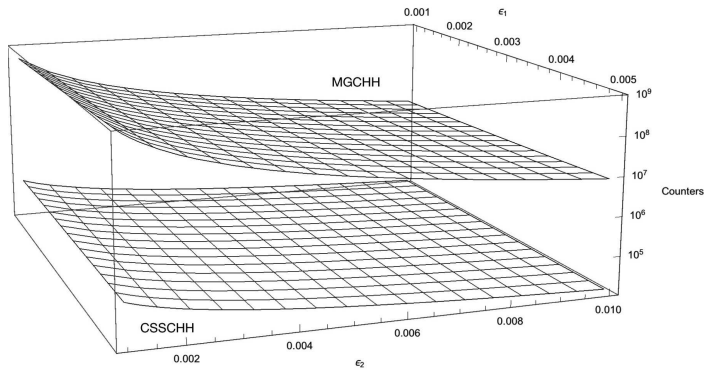} 
  \caption{Counters required by MGCHH and CSSCHH}
  \label{counters}
\end{figure}

As shown by the two surfaces, CSSCHH requires several orders of magnitude less counters than MGCHH. Moreover, we expect MGCHH to be slower with regard to CSSCHH. Indeed, every update step, in which the incoming primary stream item is not monitored and all of the $s_1$ counters are full, requires not only decrementing all of the $s_1$ counters (as in the Frequent algorithm), but also, for each one of them, MGCHH must randomly select a correlated item and decrease its frequency as well.  Now, we consider and discuss the accuracy of MGCHH. Being based on the Frequent algorithm, we know that, overall, its accuracy shall be lower than the accuracy provided by CSSCHH which is, instead, based on the Space Saving algorithm. Indeed, it is well known that Space Saving is  more accurate than Frequent \cite{Cormode,Manerikar}. In the next section we shall experimentally see how much faster and accurate CSSCHH is with regard to MGCHH.

\section{Experimental results}
\label{results}

In order to compare and evaluate our CSSCHH algorithm against MGCHH we have implemented them in C++. The source code has been compiled using the clang c++ compiler v8.0 on Mac OS X v10.12 with the following flags: -Os -std=c++14. We recall here that, on Mac OS X, the optimization flag -Os provides better optimization than the -O3 flag and is the standard for building the release build of an application. The tests have been carried out on a machine equipped wth 16 GB of RAM and a 3.2 GHz quad-core Intel Core i5 processor with 6 MB of cache level 3.

The items in the synthetic datasets used in our experiments are distributed according to the Zipf distribution. In each one of the experiments, the execution have been repeated 10 times using a different seed for the pseudo-random number generator used for creating the input data stream (using the same seeds in the corresponding executions of different algorithms). For each input distribution generated, the results have been averaged over all of the runs. The input items are 32 bits unsigned integers. Table \ref{experiments} reports all of the experiments on synthetic datasets that have been carried out.

We have also experimented using a real dataset, namely Worldcup'98. This dataset is publicly available\footnote{http://ita.ee.lbl.gov/html/contrib/WorldCup.html} and stores information related to the requests made to the World Cup web site during the 1998 tournament. For each request, the dataset includes a ClientID (which is a unique integer identifier for the client that issued the request) and an ObjectID (again, a unique integer identifier for the requested URL). In this experiment we determine correlated heavy hitters between ClientID and ObjectID pairs, treating ClientID as the primary items, and ObjectID as the secondary item. Owing to the huge size of the full dataset, we used a subset of the available data, i.e., the data from day 41 to day 46 of the competition. Table \ref{real-dataset} reports the statistical characteristics of the real dataset.

Regarding synthetic datasets, we vary the input data stream size ($n$, in millions), the skew of the zipfian distribution ($\rho$), the total space used (measured in MegaBytes) and the $\phi_1$ and $\phi_2$ support thresholds. The value of the remaining parameters are fixed and reported in each individual plot. Regarding the total space used, it is worth noting here that, for MGCHH, a counter (related to either a primary or a correlated item) requires 4 bytes to store the monitored item (an unsigned int) and 8 bytes to store its estimated frequency (a long int), for a total of 12 bytes. On the other hand, for CSSCHH a counter related to primary items requires 12 bytes as well, but a counter monitoring a tuple $(x, y)$ requires instead 4 bytes for $x$, 4 bytes for $y$ and 8 bytes to store the estimated frequency, for a total of 16 bytes. Therefore, the total space used by MGCHH is $12 (s_1 + s_1 s_2)$ bytes, whilst the total space used by CSSCHH is $12 k_1 + 16 k_2$ bytes. The test related to the space used is carried out by assigning to both algorithm exactly the same space (measured in MegaBytes) in order to fairly compare both algorithms and to understand how the algorithms behave when performing under exactly the same conditions (with regard to the space used). Therefore, the allocated space determines the number of counters to be used correspondingly by the algorithms. In all of the other tests, we preserved this property choosing $s_1$, $s_2$, $k_1$ and $k_2$ such that $12k_1 + 16 k_2 =  12 (s_1 + s_1 s_2)$ and $k_1=s_1$. Table \ref{counters-space} reports the counters corresponding to the space used in the experiments.  

We begin our analysis discussing the results for synthetic datasets. The \emph{recall} is the total number of true frequent items reported over the number of frequent items given by an exact algorithm. Therefore, an algorithm is correct iff its \emph{recall} is equal to one. Since the algorithms under test are based respectively on Frequent (MGCHH) and on Space Saving (CSSCHH), their \emph{recall} is always one if they are allowed to use enough counters. In all of the test we used a number of counters $s_1$ and $k_1$ greater than $\frac{1}{\phi_1}$ for the primary items and a number of counters for correlated items greater that $\frac{1}{\phi_2}$ for $s_2$ and greater than $\frac{1}{\phi_1 \phi_2}$ for $k_2$. This guarantees both algorithms to reach a \emph{recall} equal to one in every case. The plots on the \emph{recall} have not been reported in the paper.

Next, we analyze the accuracy, beginning with the \emph{precision} attained (with regard to the CHHs). Since \emph{precision} is defined as the total number of true frequent items reported over the total number of items reported, this metric quantifies the number of false positives outputted by an algorithm. It follows that, from this point of view, the algorithm's \emph{precision} should ideally be one. As shown in Figure~\ref{precision}, CSSCHH clearly outperforms MGCHH with regard to the \emph{precision}. Indeed, CSSCHH is consistently able to provide one or near one \emph{precision} in all of the tests carried out, whilst MGCHH lags far behind.

Accuracy is also related to the absolute and relative errors on the frequency estimate committed by the algorithms. Denoting with $f$ the exact frequency of a CHH and with $\hat{f}$ the corresponding frequency reported by an algorithm, then the absolute error is, by definition, the difference $\left|  f - \hat{f} \right|$. 

Similarly, the relative error is defined as $\frac{{\left| {f - \hat{f}} \right|}}{f}$ and the average relative error is derived by averaging the relative errors over all of the measured frequencies. Figures~\ref{Abs_error}~and~\ref{Rel_error} depict respectively absolute and relative errors committed by the algorithms (with regard to CHHs). Each single plot reports both the maximum and the mean values attained by both algorithms. Again, it is immediate verifying that CSSCHH is extremely accurate in all of the cases, with absolute and relative errors always equal to zero. On the contrary, MGCHH estimates are clearly affected by significant error.

Finally, we evaluated the algorithm with regard to their speed in processing stream items. Figure \ref{Updates} shows the speed attained, reported as updates per millisecond. Again, CSSCHH outperforms MGCHH, being consistently much faster in all of the cases. In particular, MGCHH is more than three times faster in all of the tests that have been carried out, except the tests in which we vary the skew of the input distribution and the space we allow to be used. Anyway, as shown by the plots, CSSCHH is always faster than MGCHH.

Next, we compare MGCHH and CSSCHH using the real dataset Worldcup'98. Figures \ref{wc-space} and \ref{wc-phi2} depict  the results obtained respectively when varying the space allowed and the $\phi_2$ threshold (fixing the $\phi_1$ threshold). Table \ref{counters-space-worldcup} reports the counters corresponding to the space used. It is worth noting here that, owing to the skewness of the dataset being really low (see Table \ref{real-dataset}), we fixed $\phi_1$ to a low value ($0.001$) and  $\phi_2$ varying from $0.001$ to $0.01$, since otherwise no (or just a few) correlated heavy hitters exist. As shown, both algorithms exhibit the same behaviour already observed on synthetic datasets, with CSSCHH clearly outperforming MGCHH with regard to every metric under consideration.

We conclude by noting that the experimental results fully confirm our theoretical expectations reported in the previous section. CSSCHH is more accurate in terms of both precision, absolute and relative error committed. Our algorithm is also faster than MGCHH. Therefore, CSSCHH is a better alternative to MGCHH for mining correlated heavy hitters.

\begin{table}
\renewcommand{\arraystretch}{1.3}
 \caption{Experiments carried out (synthetic datasets)}
      \label{experiments}
	\centering
    \begin{tabular}{| c | c | c | c |}
    \hline
      Items ($n$, millions) & Skew ($\rho$) & Space used (MBytes) & Threshold ($\phi_1$) \\
      \hline
    \{5, 50, 500, 1000\}  & \{1, 1.4, 1.8, 2.2\} & \{1.009, 3.941, 15.573, 61.908\} & \{0.1, 0.01, 0.001\} \\ \hline
       
        \end{tabular}
    \end{table}

\begin{table}
\renewcommand{\arraystretch}{1.3}
 \caption{Counters corresponding to the space used (synthetic datasets)}
      \label{counters-space}
	\centering
    \begin{tabular}{| c | c | c | c |}
    \hline
      Space (MB) & $k_1$ = $s_1$ & $k_2$ & $s_2$ \\
      \hline
      1.009 & 4,200 & 63,000 & 20 \\ \hline
      3.941 & 8,400 & 252,000 & 40 \\ \hline
      15.573 & 16,800 & 1,008,000 & 80 \\ \hline
      61.908 & 33,600 & 4,032,000 & 160 \\ \hline
       
        \end{tabular}
    \end{table}

\begin{table}
\renewcommand{\arraystretch}{1.3}
 \caption{Statistical characteristics of the real dataset (Worldcup'98)}
      \label{real-dataset}
	\centering
    \begin{tabular}{| c | c | c | c | c | c | c | c | c |}
    
    \hline  & primary & secondary \\ \hline
    Count & 104,271,758 & 104,271,758  \\ \hline
    Distinct items & 539,464 & 21,605 \\ \hline
    Min & 1 & 0 \\ \hline 
    Max & 1,375,004 &  40,317 \\ \hline 
    Mean &  549,924 & 9,774.51 \\ \hline
    Median & 578,870 & 887 \\ \hline
    Std. deviation &  427,859 & 11,094.8 \\ \hline
    Skewness & 0.146792 & 0.449196 \\ \hline
   
    \end{tabular}
    \end{table}

\begin{table}
\renewcommand{\arraystretch}{1.3}
 \caption{Counters corresponding to the space used (Worldcup'98)}
      \label{counters-space-worldcup}
	\centering
    \begin{tabular}{| c | c | c | c |}
    \hline
      Space (MB) & $k_1$ = $s_1$ & $k_2$ & $s_2$ \\
      \hline
      11.558 & 10,000 & 750,000 & 100 \\ \hline
      36.048 & 25,000 & 2,343,750 & 125 \\ \hline
      86.402 & 50,000 & 5,625,000 & 150 \\ \hline
      230.026 & 100,000 & 15,000,000 & 200 \\ \hline
       
        \end{tabular}
    \end{table}

\begin{figure}[]
  \centering
  \begin{tabular}{ccc}
  
  \subfloat[varying $n$]{
           \includegraphics[width=0.3\textwidth]{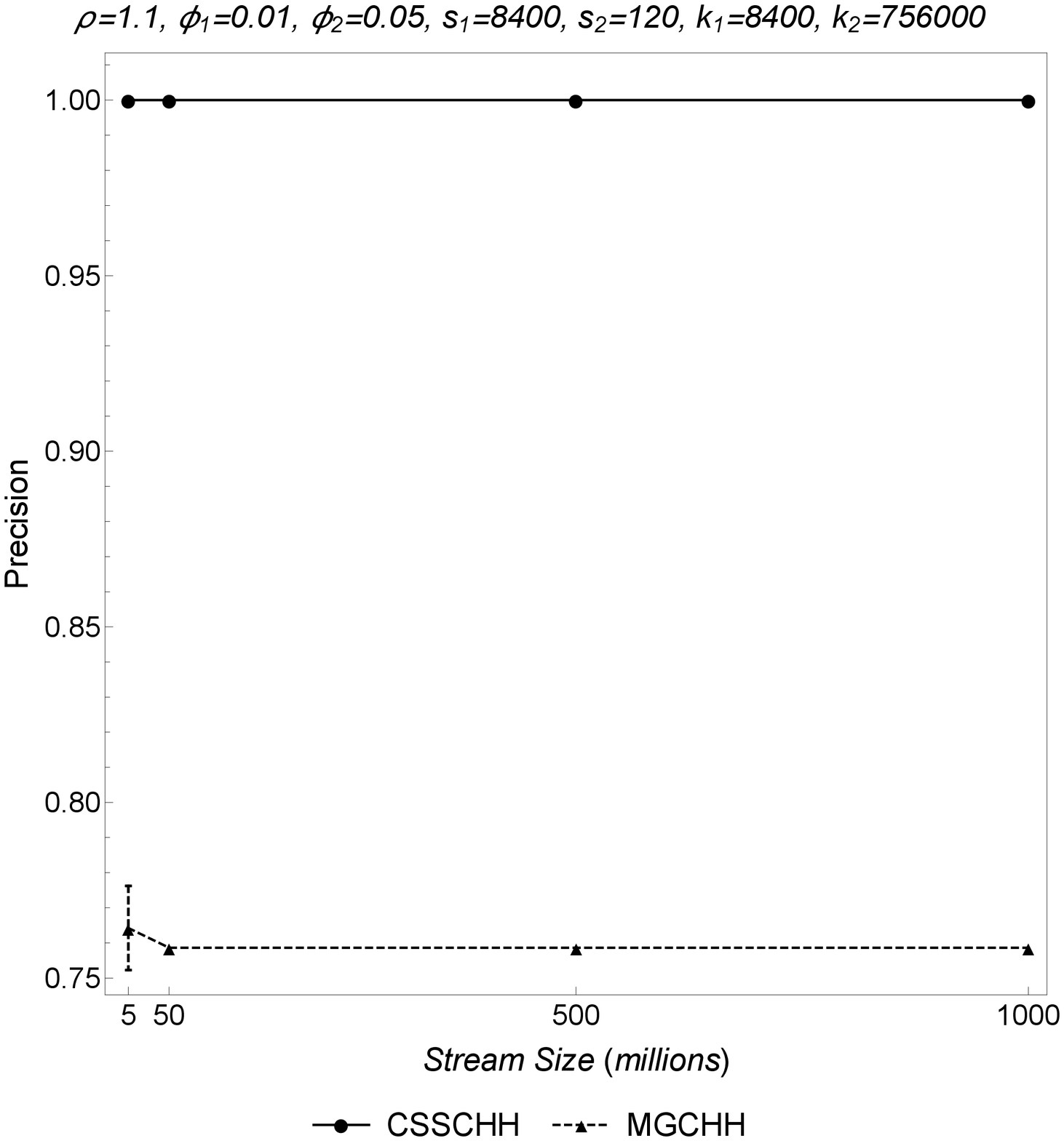}
           \label{ni-prec}
        } &

      \subfloat[varying $\rho$]{
           \includegraphics[width=0.3\textwidth]{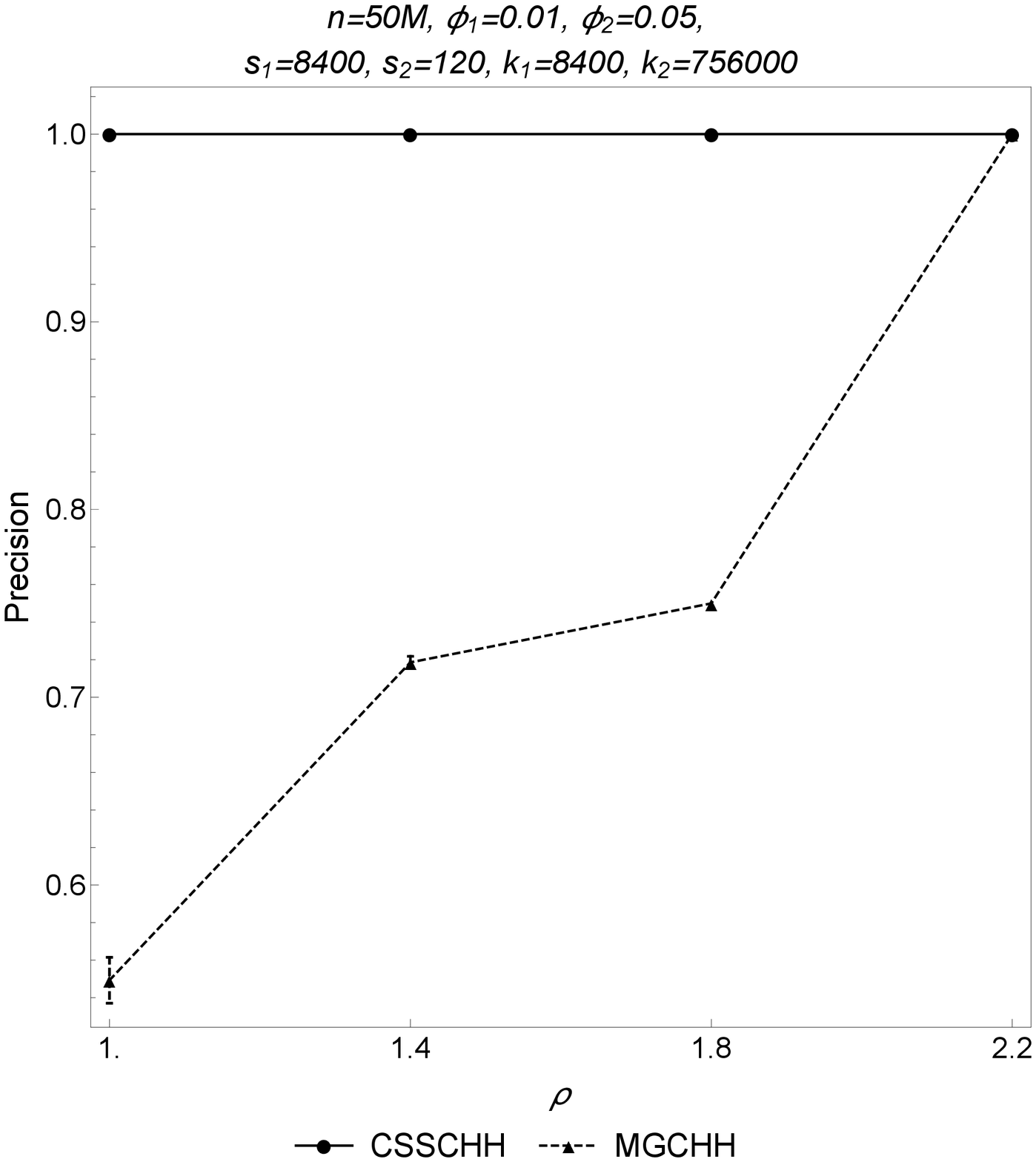}
           \label{sk-prec}
        } & 
        
      \subfloat[varying the space used ]{
           \includegraphics[width=0.3\textwidth]{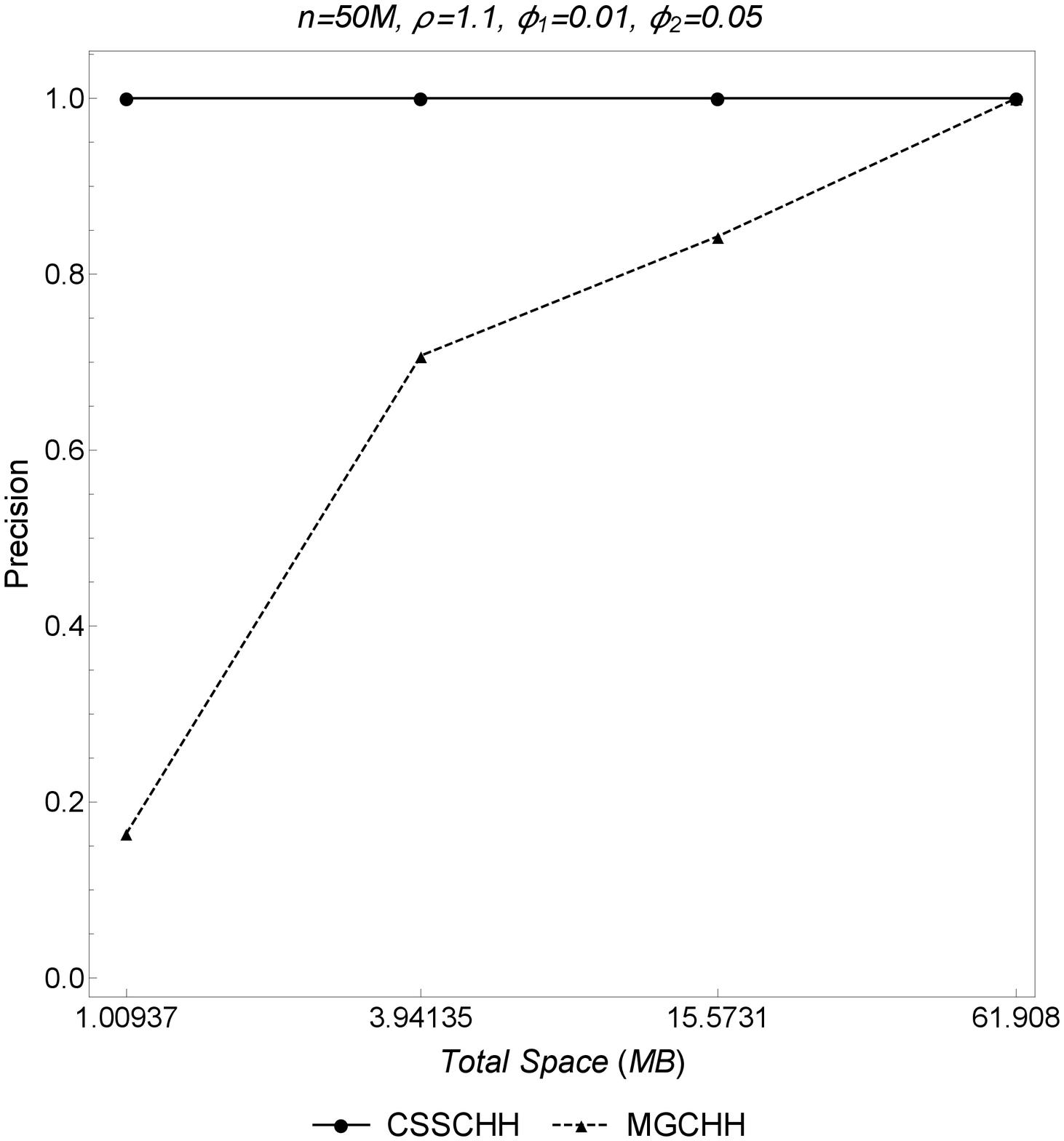}
           \label{size-prec}
        } \\

        \subfloat[$\phi_1 = 0.1$]{
           \includegraphics[width=0.3\textwidth]{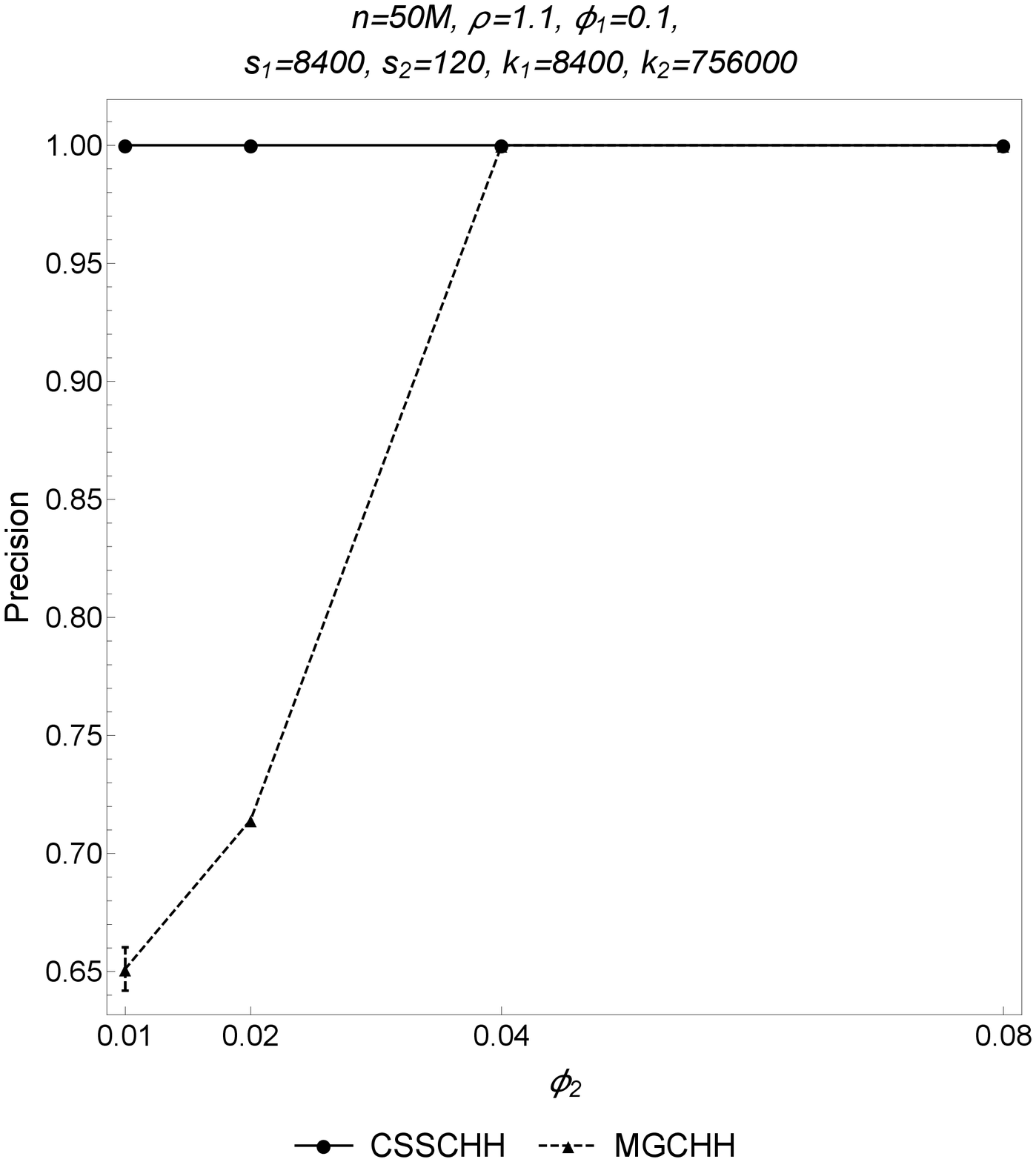}
           \label{phi1-prec1}
        } &

      \subfloat[$\phi_1 = 0.01$]{
           \includegraphics[width=0.3\textwidth]{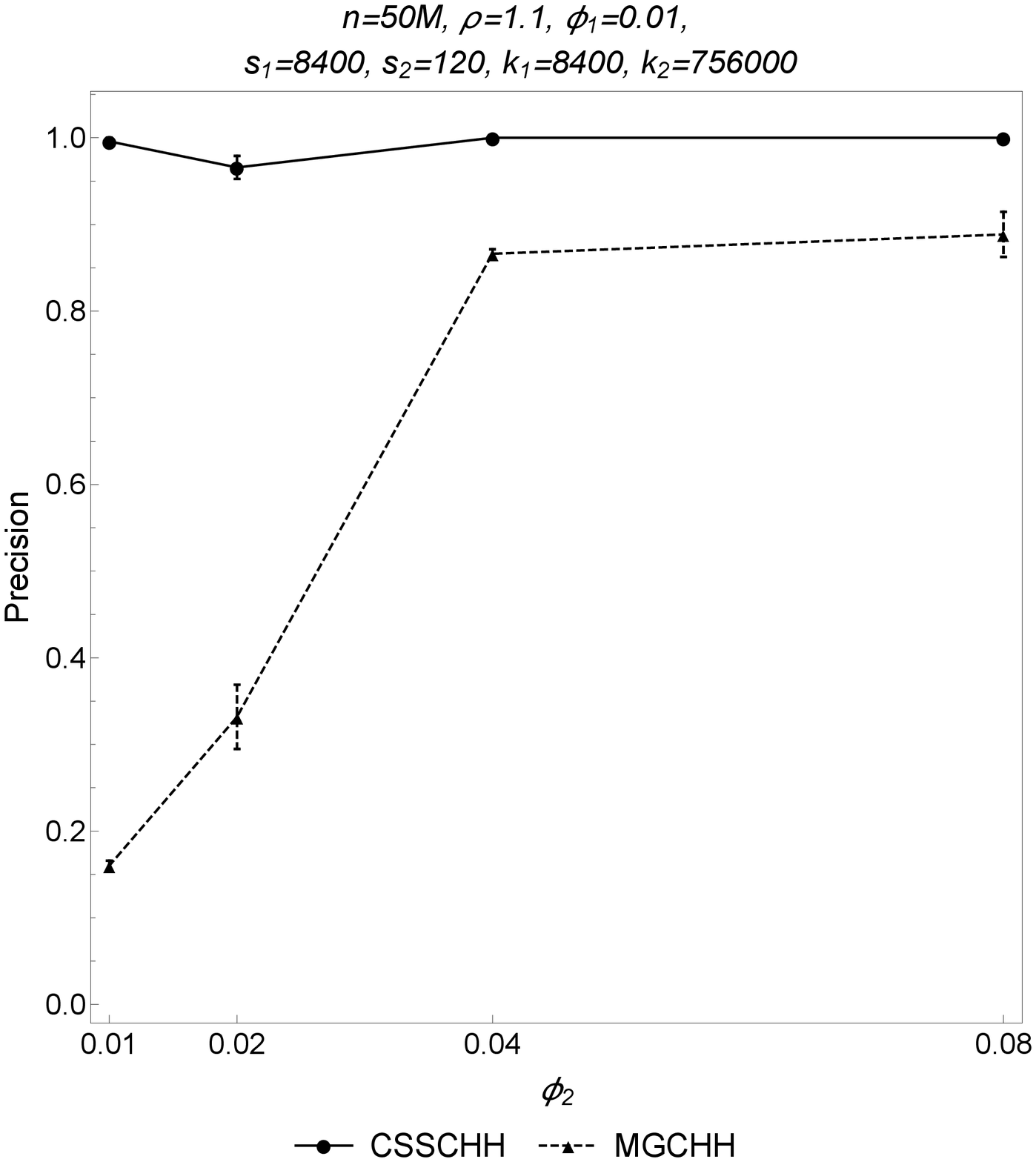}
           \label{phi2-prec2}
        } &
        
      \subfloat[$\phi_1 = 0.001$]{
           \includegraphics[width=0.3\textwidth]{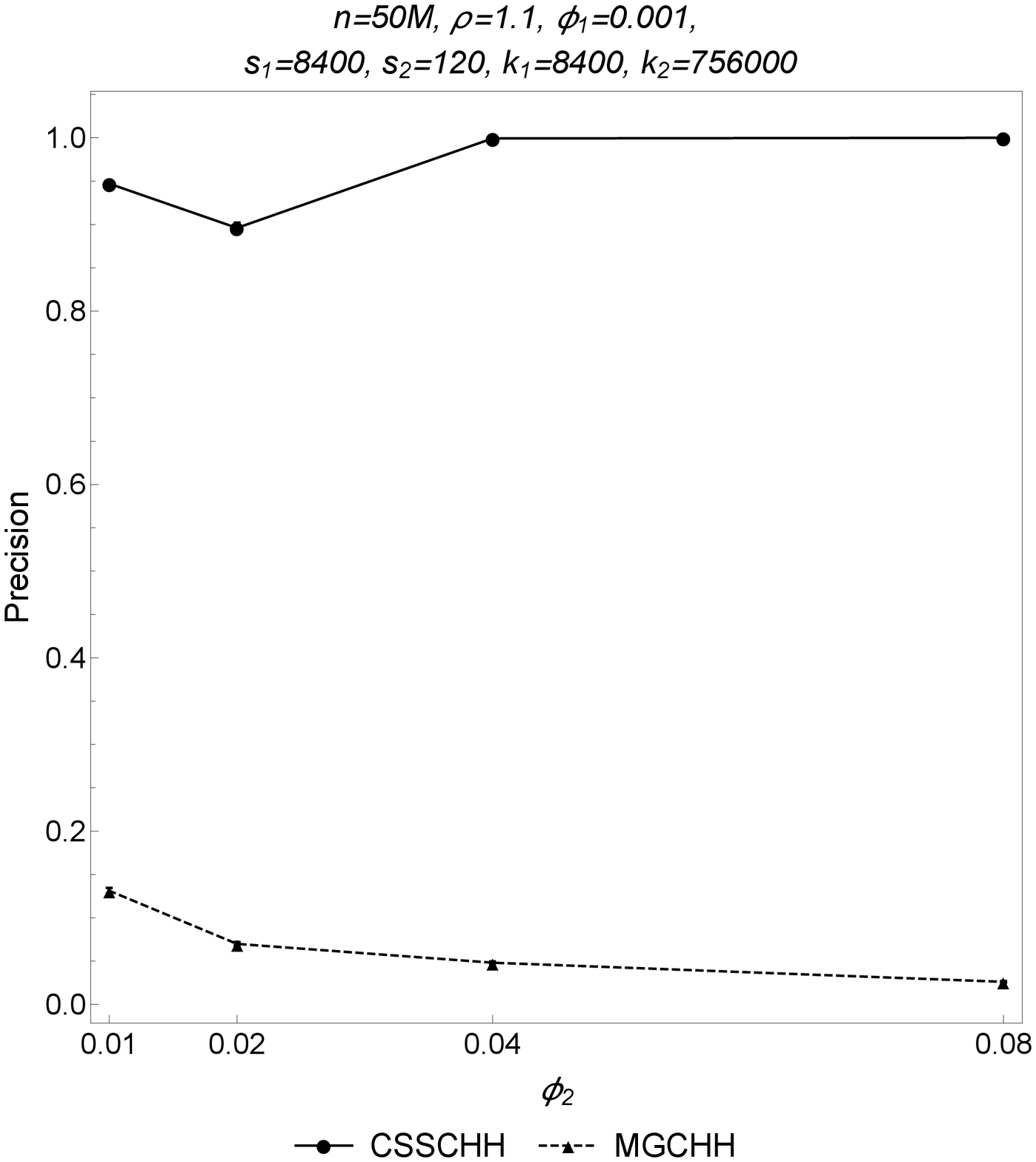}
           \label{phi3-prec3}
        } 
\end{tabular}
 
 \caption{Precision (mean and confidence interval)} 
 \label{precision}
\end{figure}

\begin{figure}[]
  \centering
  \begin{tabular}{ccc}
  
  \subfloat[varying $n$]{
           \includegraphics[width=0.3\textwidth]{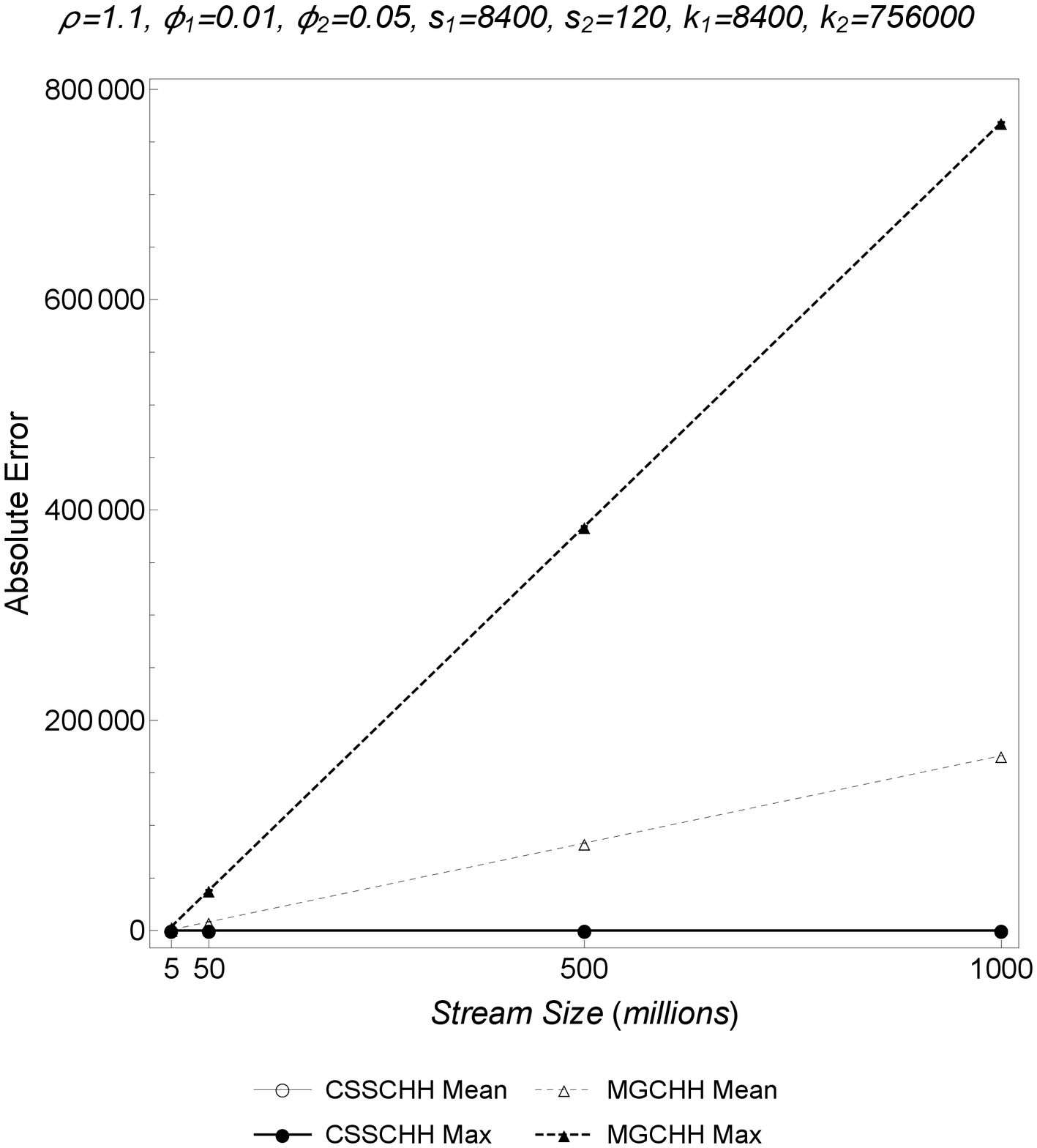}
           \label{ni-abserr}
        } &

      \subfloat[varying $\rho$]{
           \includegraphics[width=0.3\textwidth]{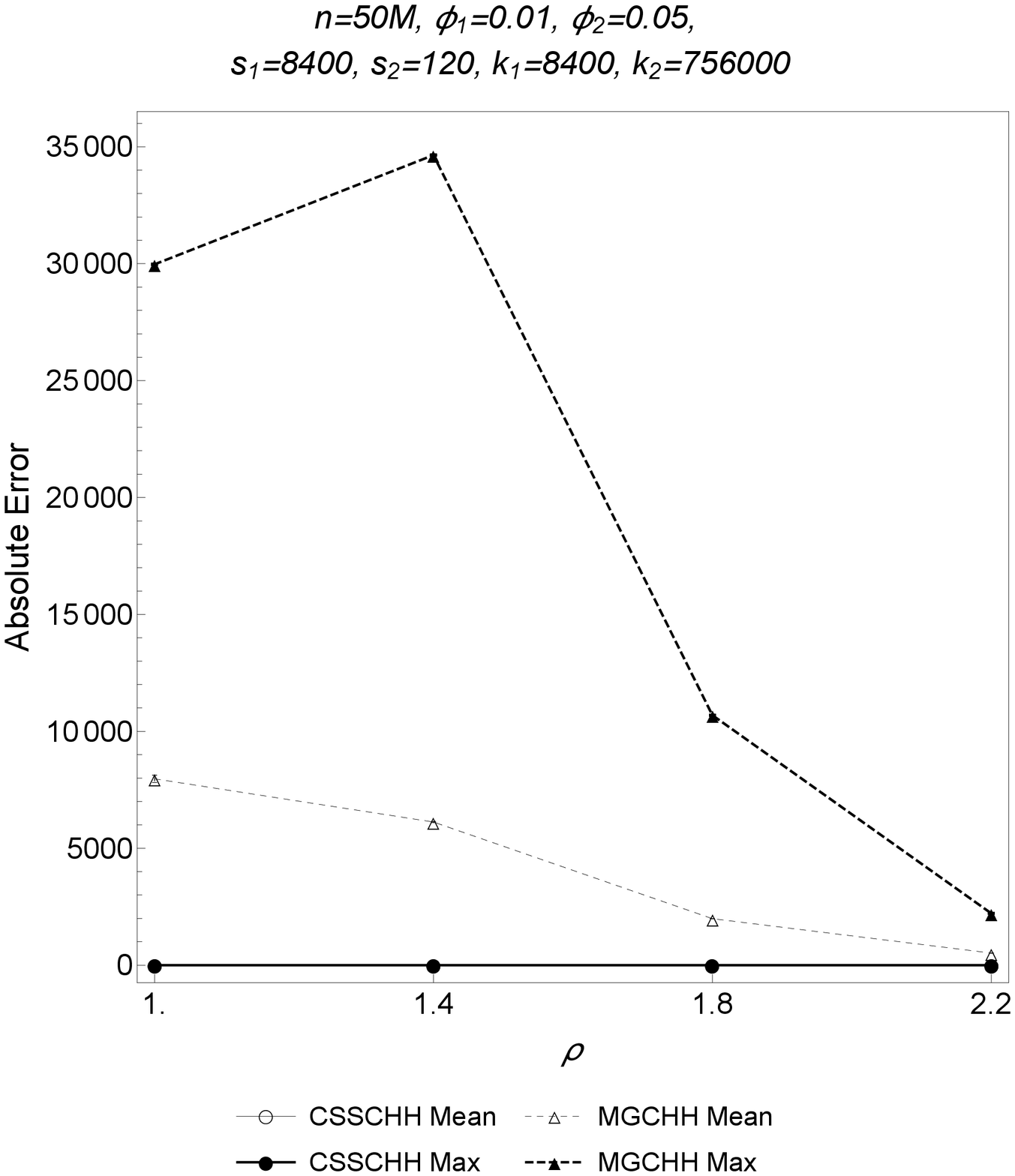}
           \label{sk-abserr}
        } &
        
      \subfloat[varying the space used ]{
           \includegraphics[width=0.3\textwidth]{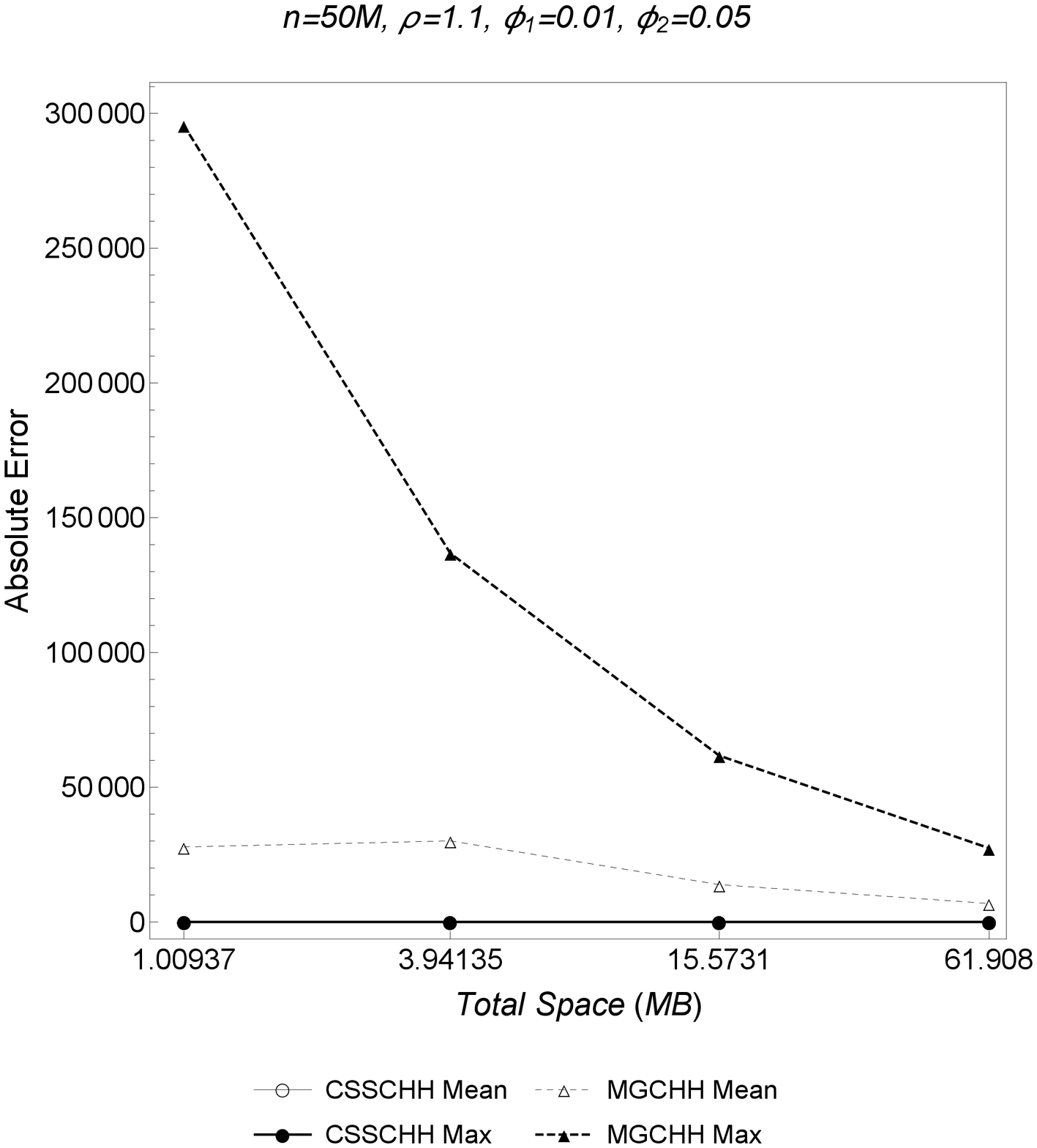}
           \label{size-abserr}
        } \\

        \subfloat[$\phi_1 = 0.1$]{
           \includegraphics[width=0.3\textwidth]{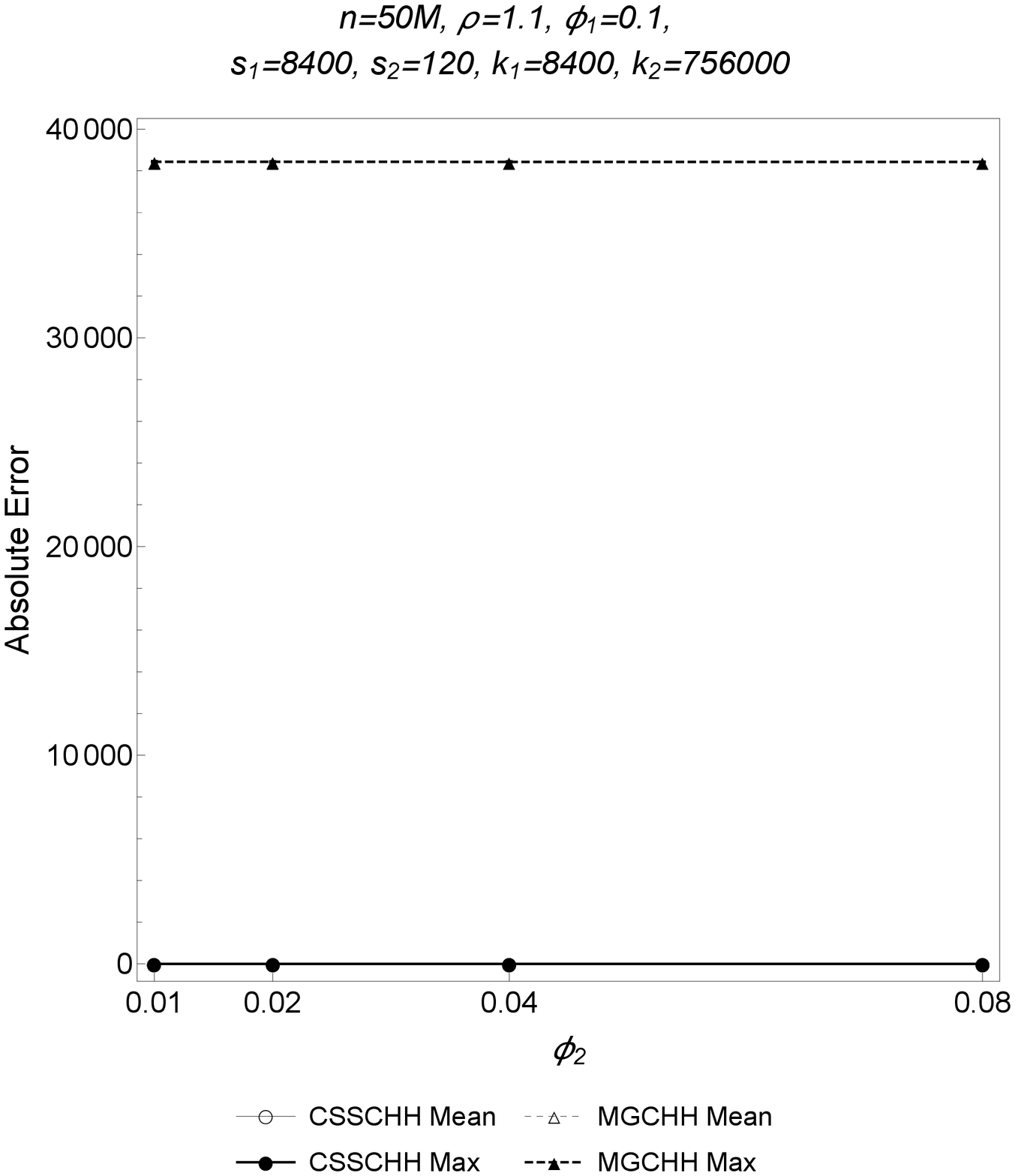}
           \label{phi1-abserr}
        }  &

      \subfloat[$\phi_1 = 0.01$]{
           \includegraphics[width=0.3\textwidth]{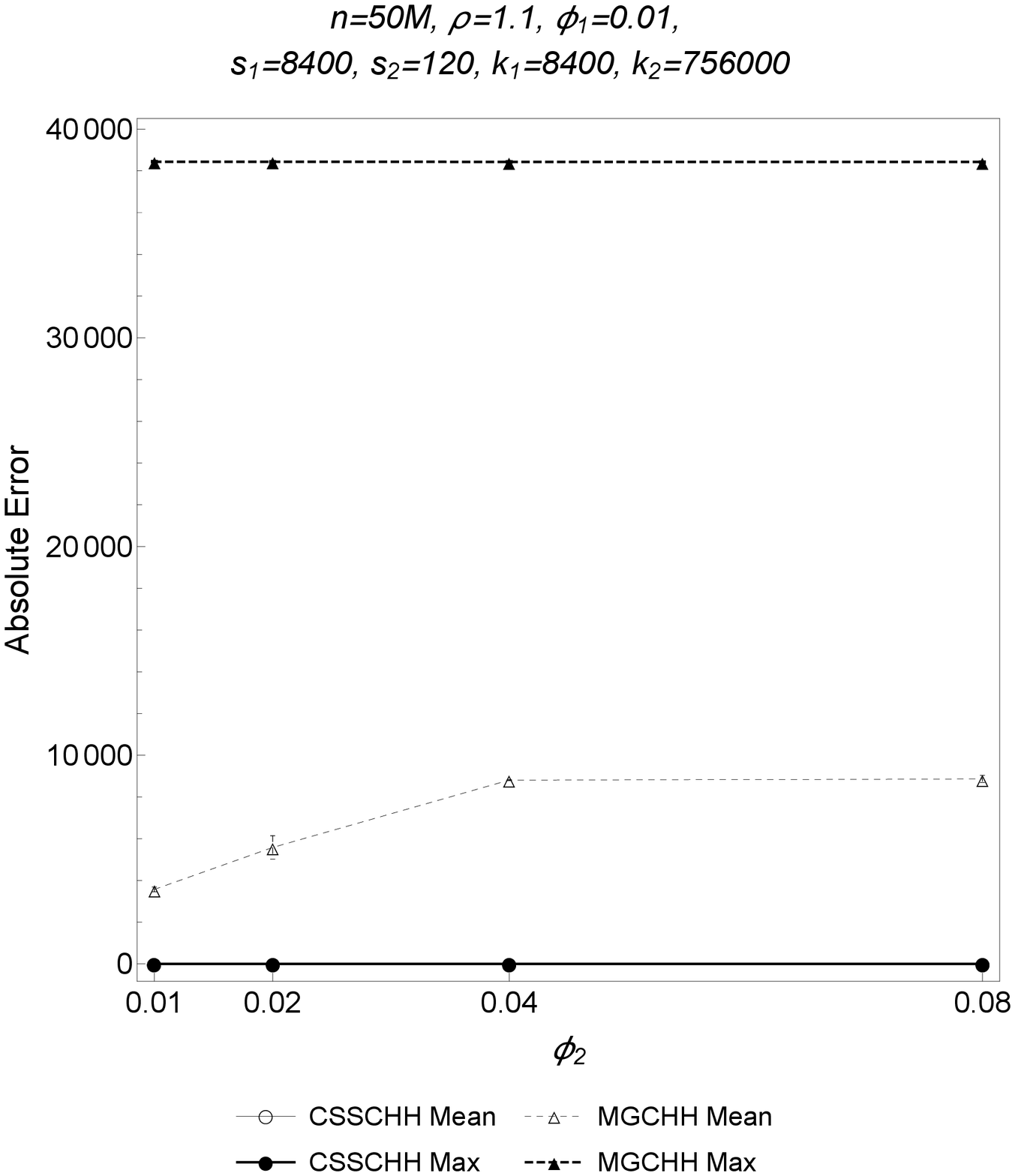}
           \label{phi2-abserr}
        } &
        
      \subfloat[$\phi_1 = 0.001$]{
           \includegraphics[width=0.3\textwidth]{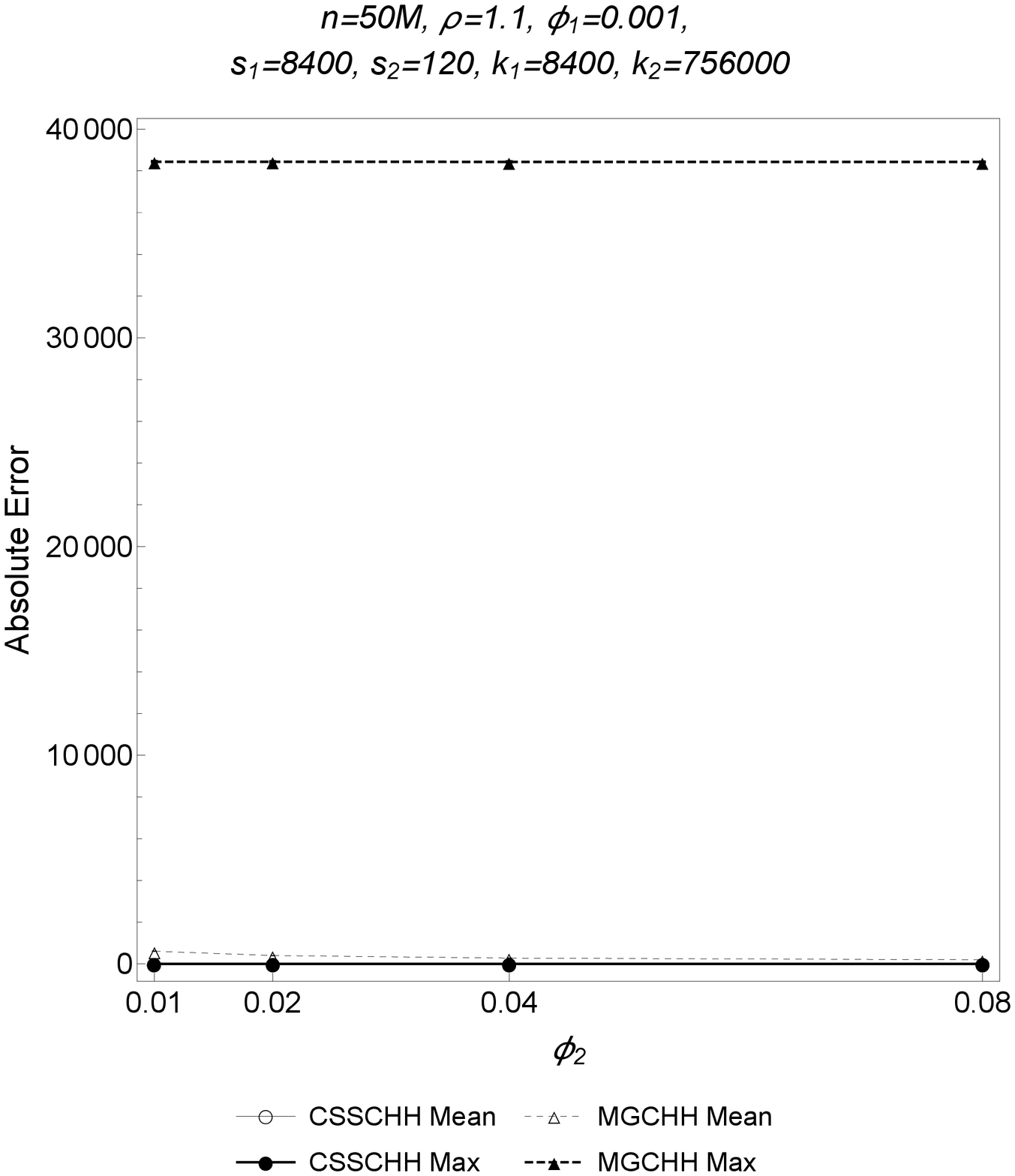}
           \label{phi3-abserr}
        } 
\end{tabular}
 
 \caption{Absolute Error (mean and confidence interval)} 
 \label{Abs_error}
\end{figure}

\begin{figure}[]
  \centering
  \begin{tabular}{ccc}
  
  \subfloat[varying $n$]{
           \includegraphics[width=0.3\textwidth]{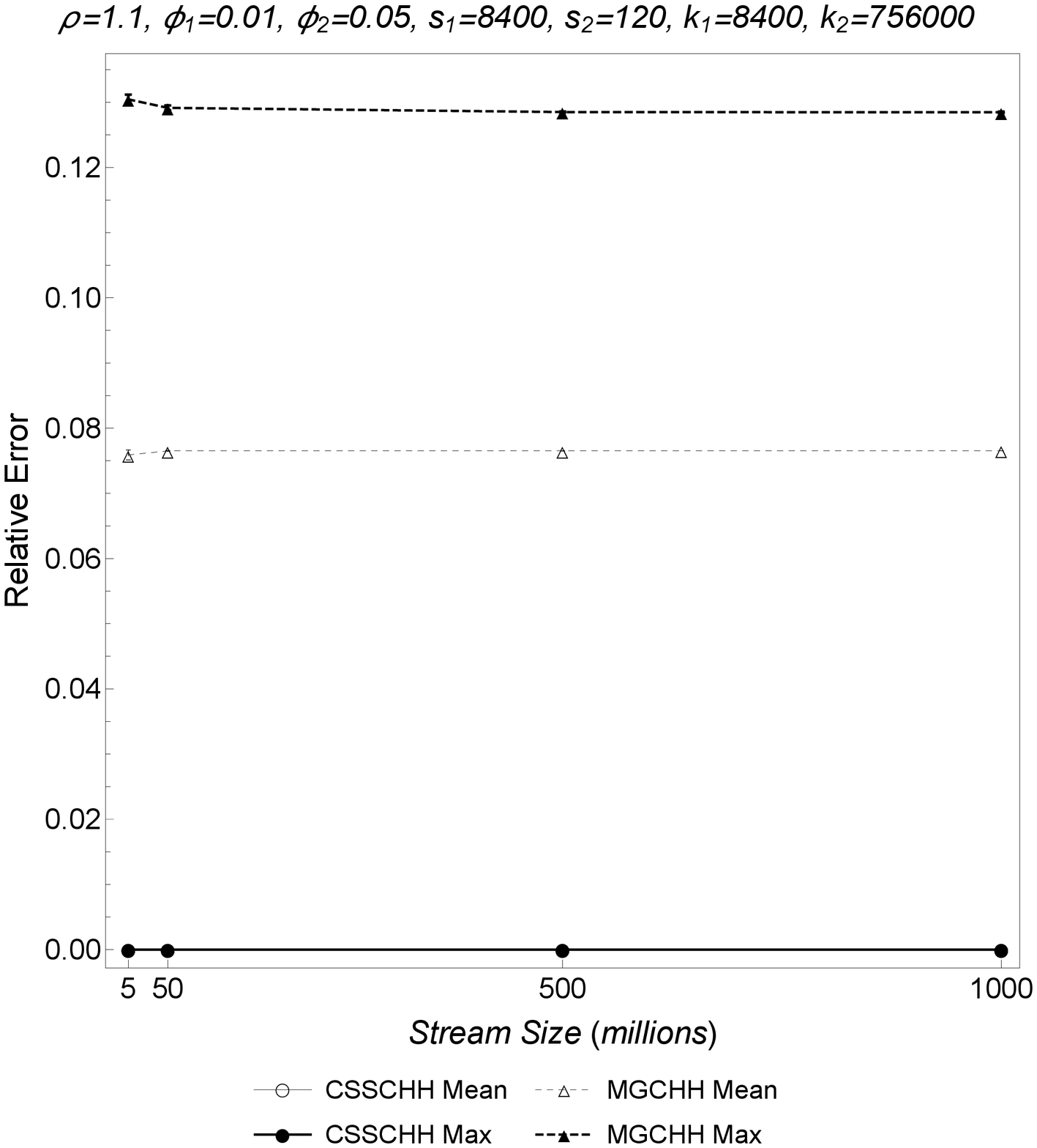}
           \label{ni-relerr}
        } &

      \subfloat[varying $\rho$]{
           \includegraphics[width=0.3\textwidth]{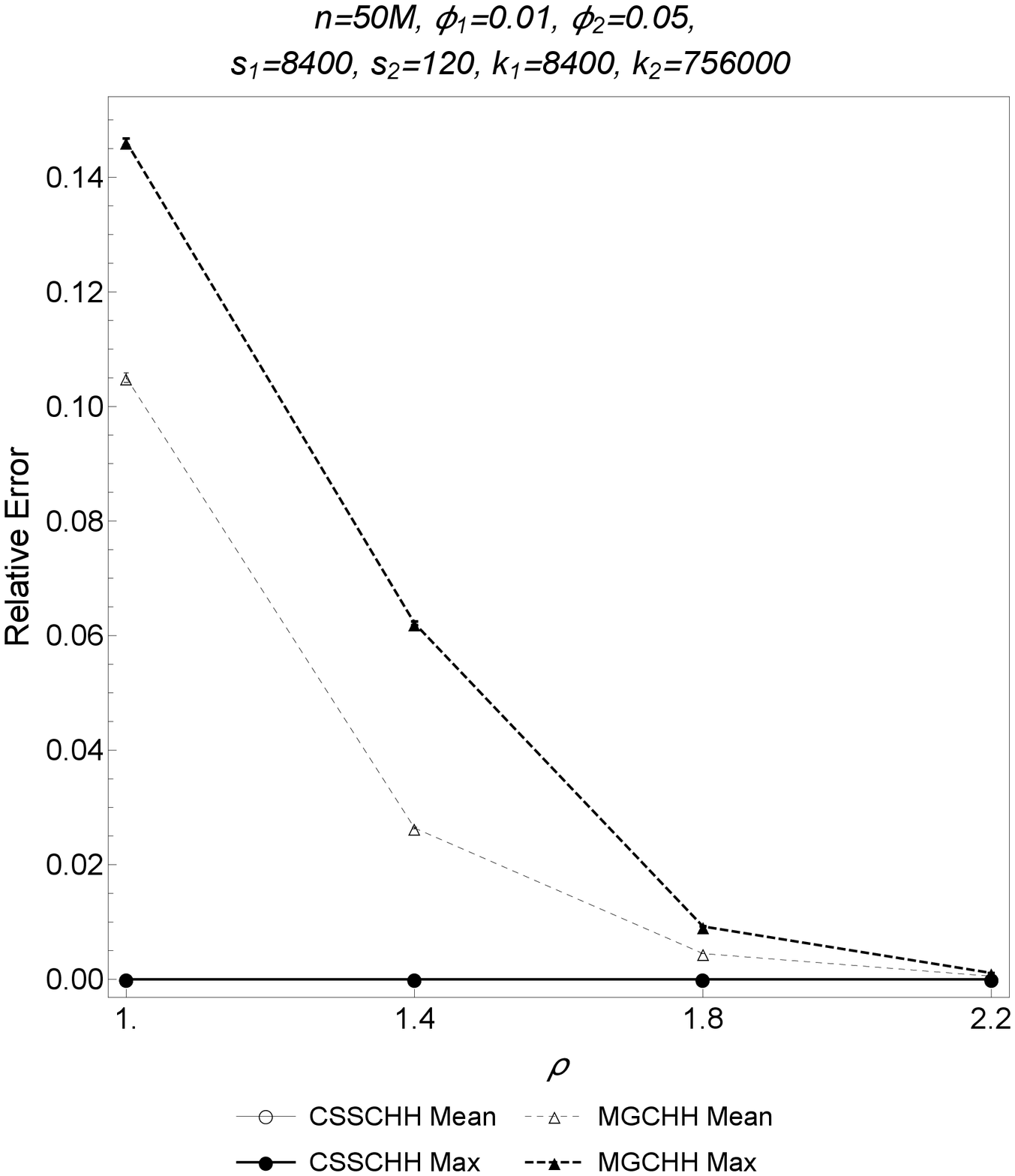}
           \label{sk-relerr}
        } &
        
      \subfloat[varying the space used ]{
           \includegraphics[width=0.3\textwidth]{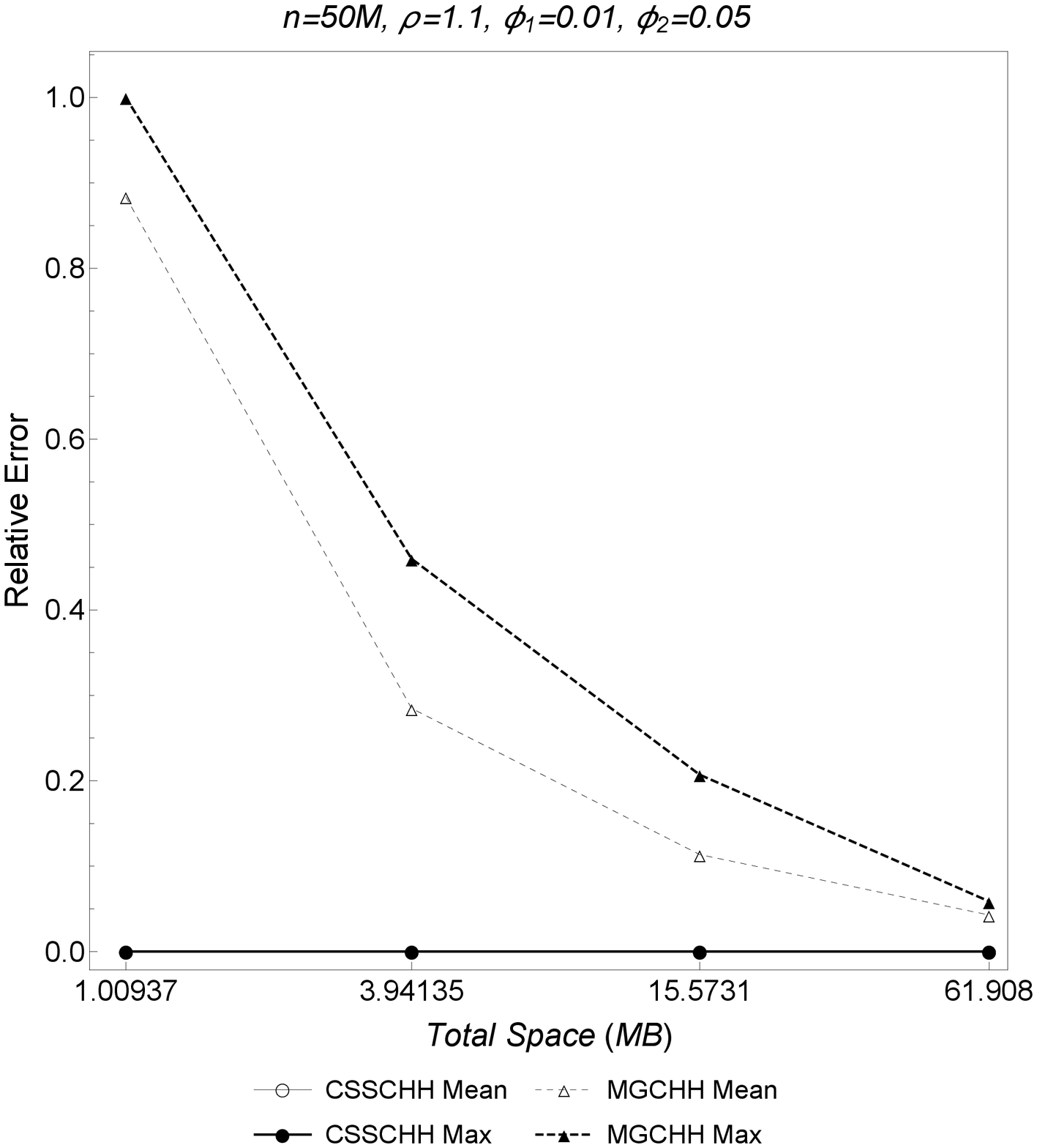}
           \label{size-relerr}
        } \\

        \subfloat[$\phi_1 = 0.1$]{
           \includegraphics[width=0.3\textwidth]{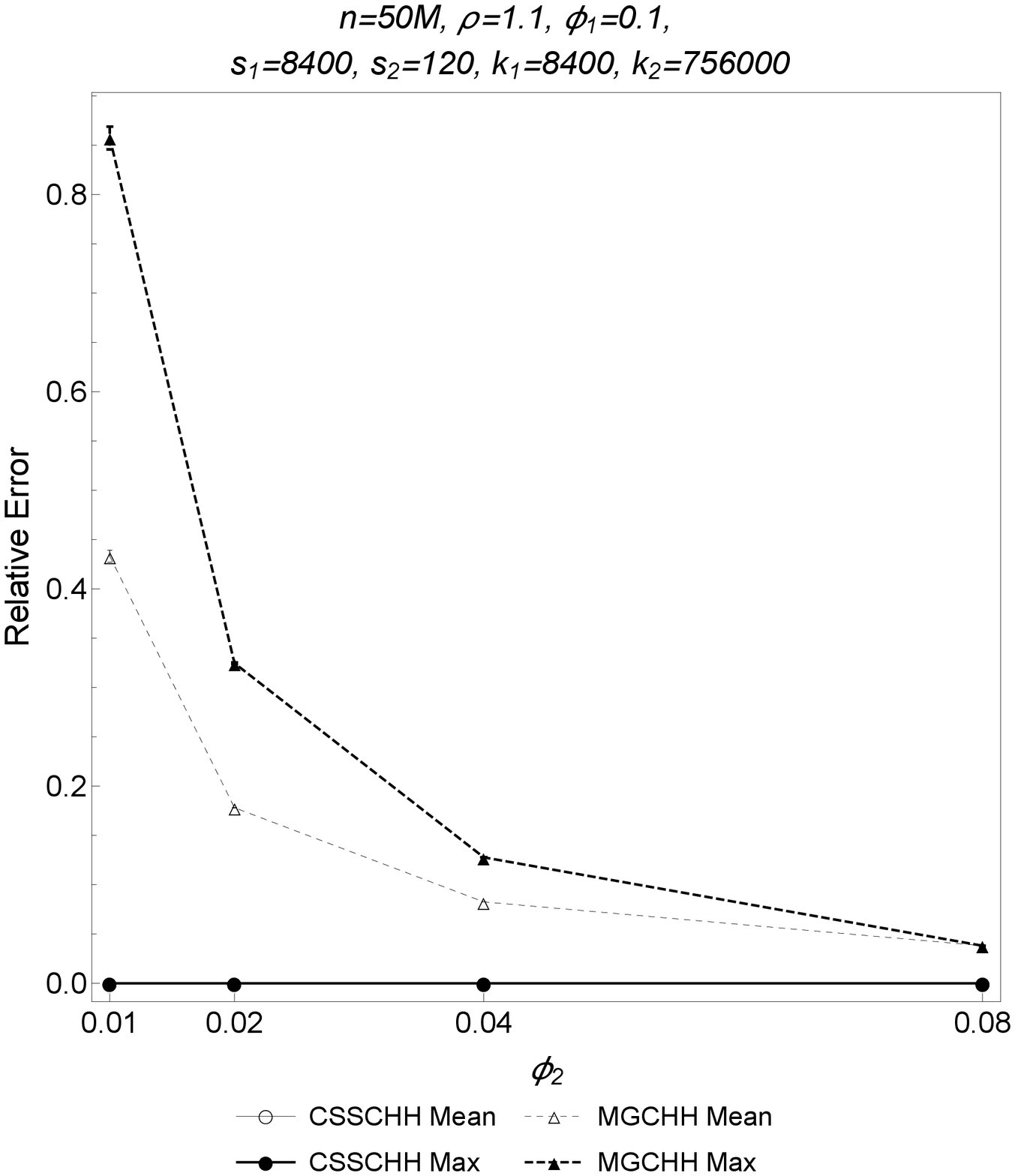}
           \label{phi1-relerr}
        } &

      \subfloat[$\phi_1 = 0.01$]{
           \includegraphics[width=0.3\textwidth]{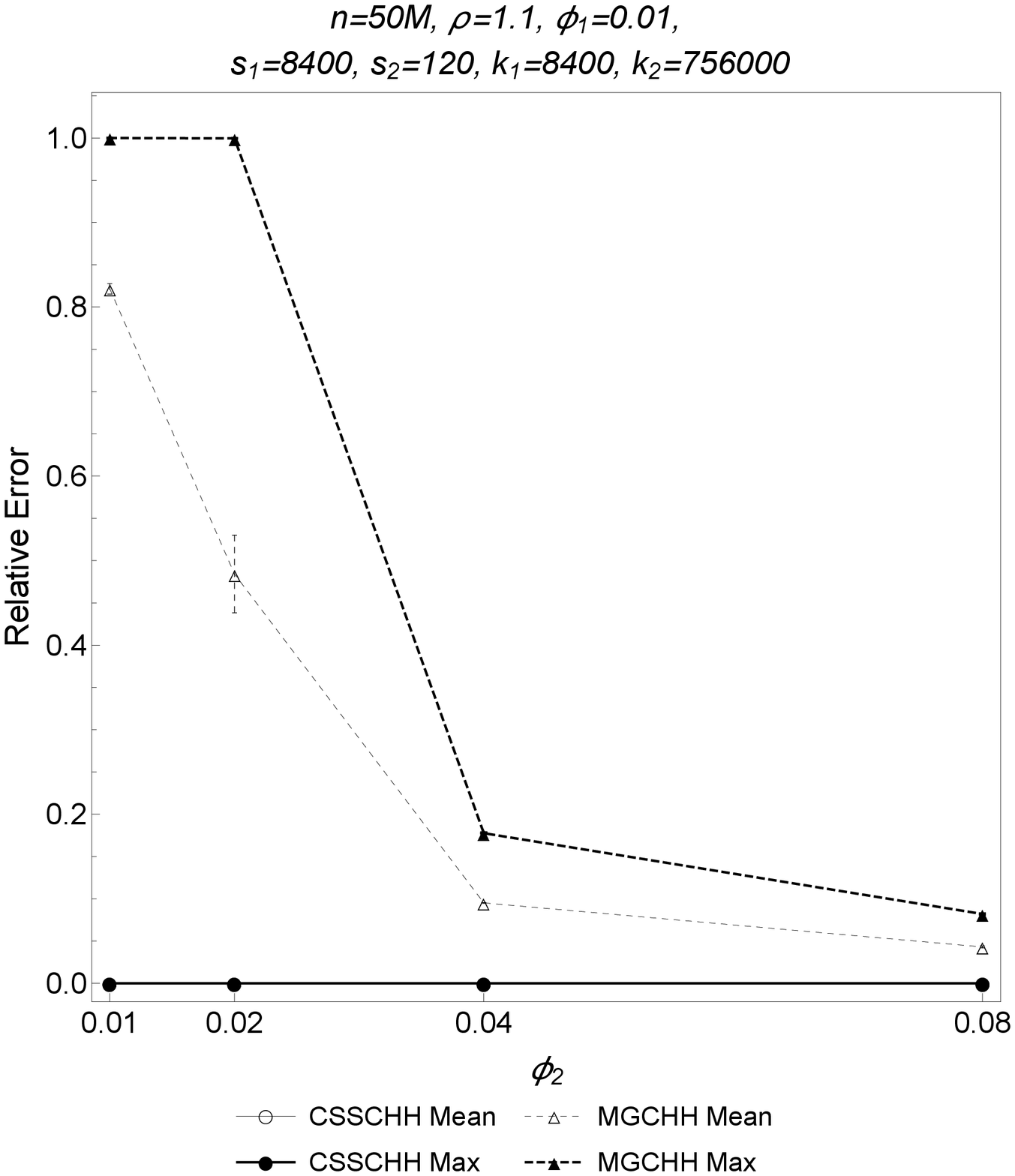}
           \label{phi2-relerr}
        } &
        
      \subfloat[$\phi_1 = 0.001$]{
           \includegraphics[width=0.3\textwidth]{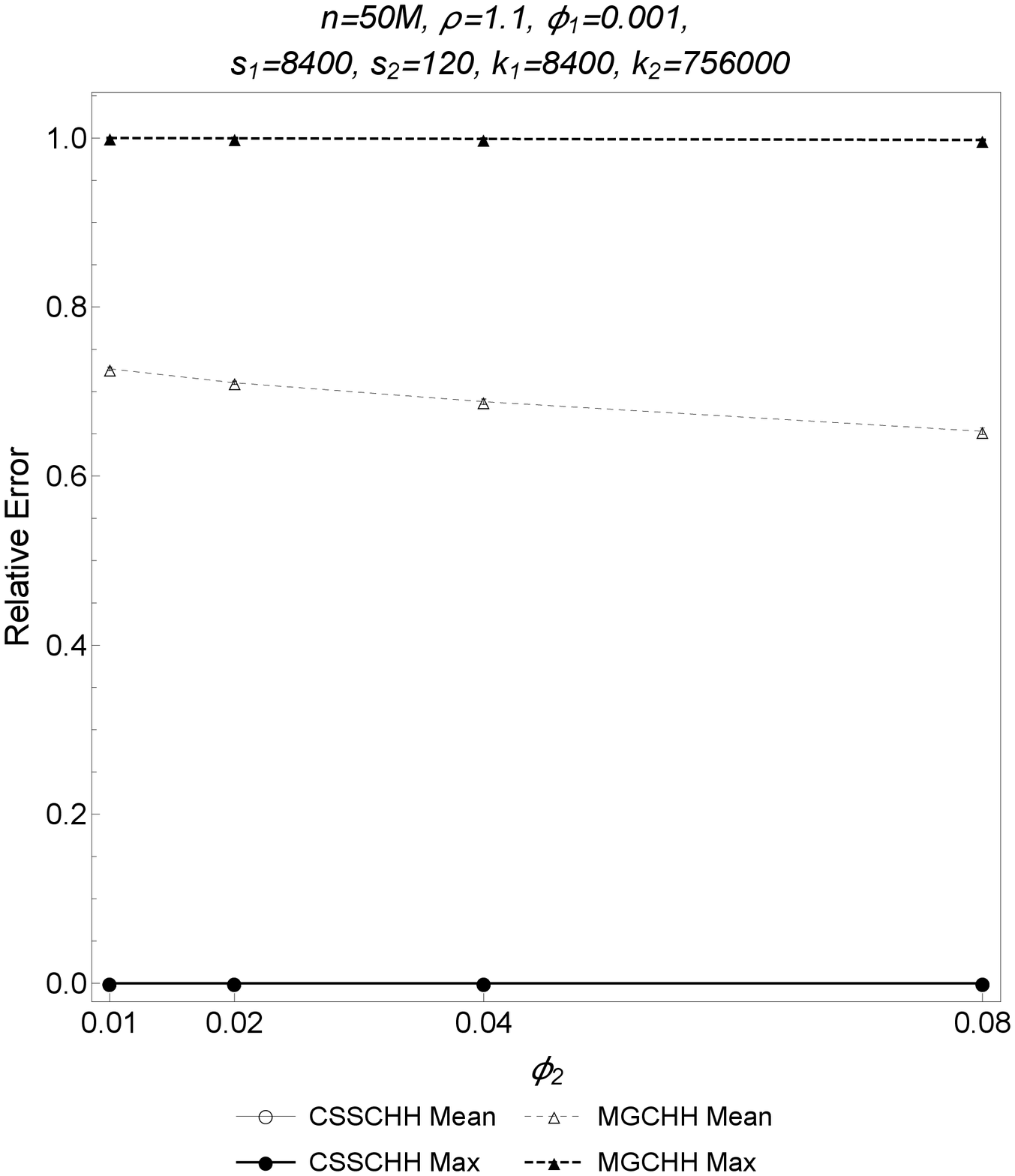}
           \label{phi3-Relerr}
        } 
\end{tabular}
 
 \caption{Relative Error (mean and confidence interval)} 
 \label{Rel_error}
\end{figure}

\begin{figure}[]
  \centering
  \begin{tabular}{ccc}
  
  \subfloat[varying $n$]{
           \includegraphics[width=0.3\textwidth]{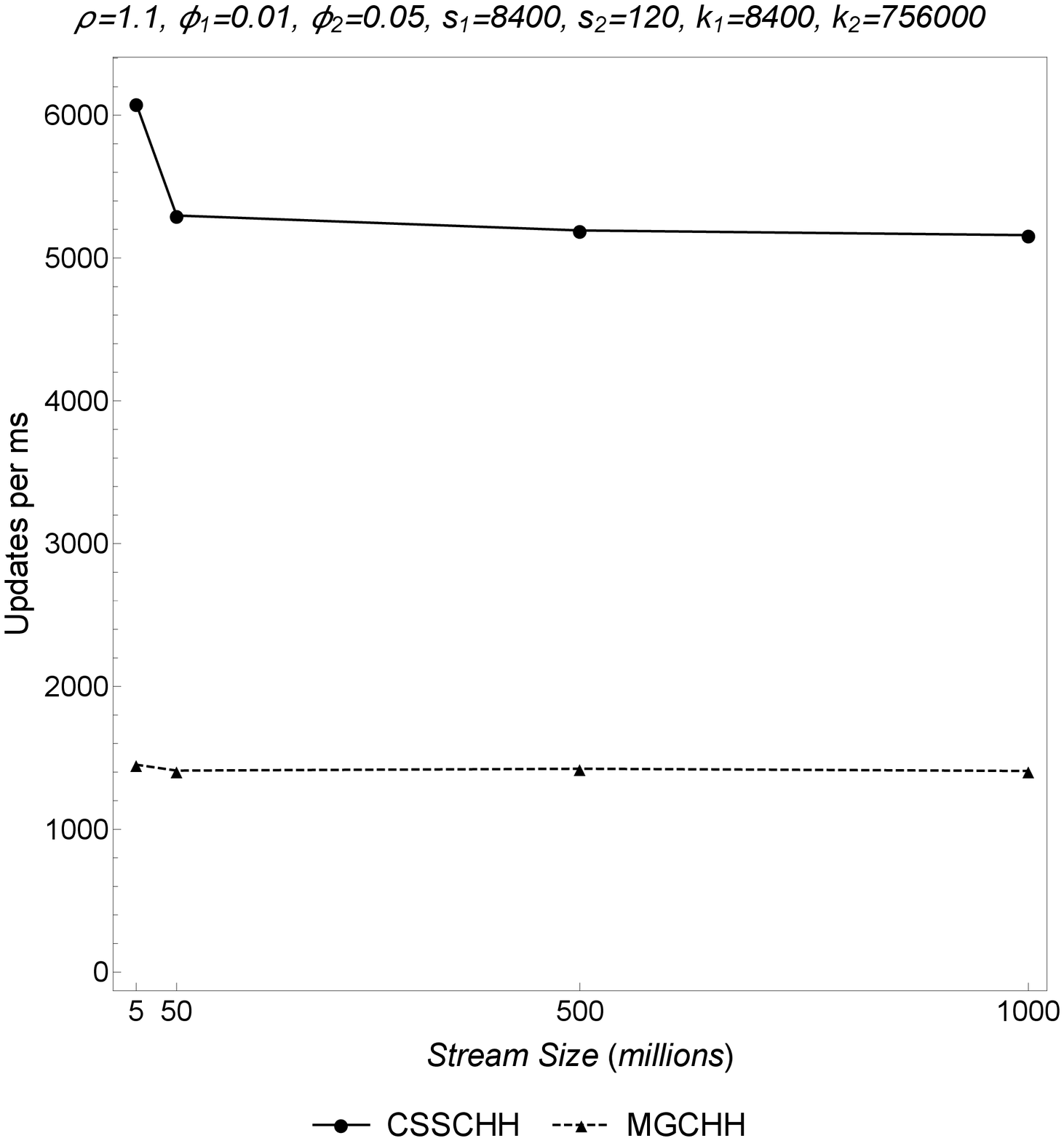}
           \label{ni-update}
        } &

      \subfloat[varying $\rho$]{
           \includegraphics[width=0.3\textwidth]{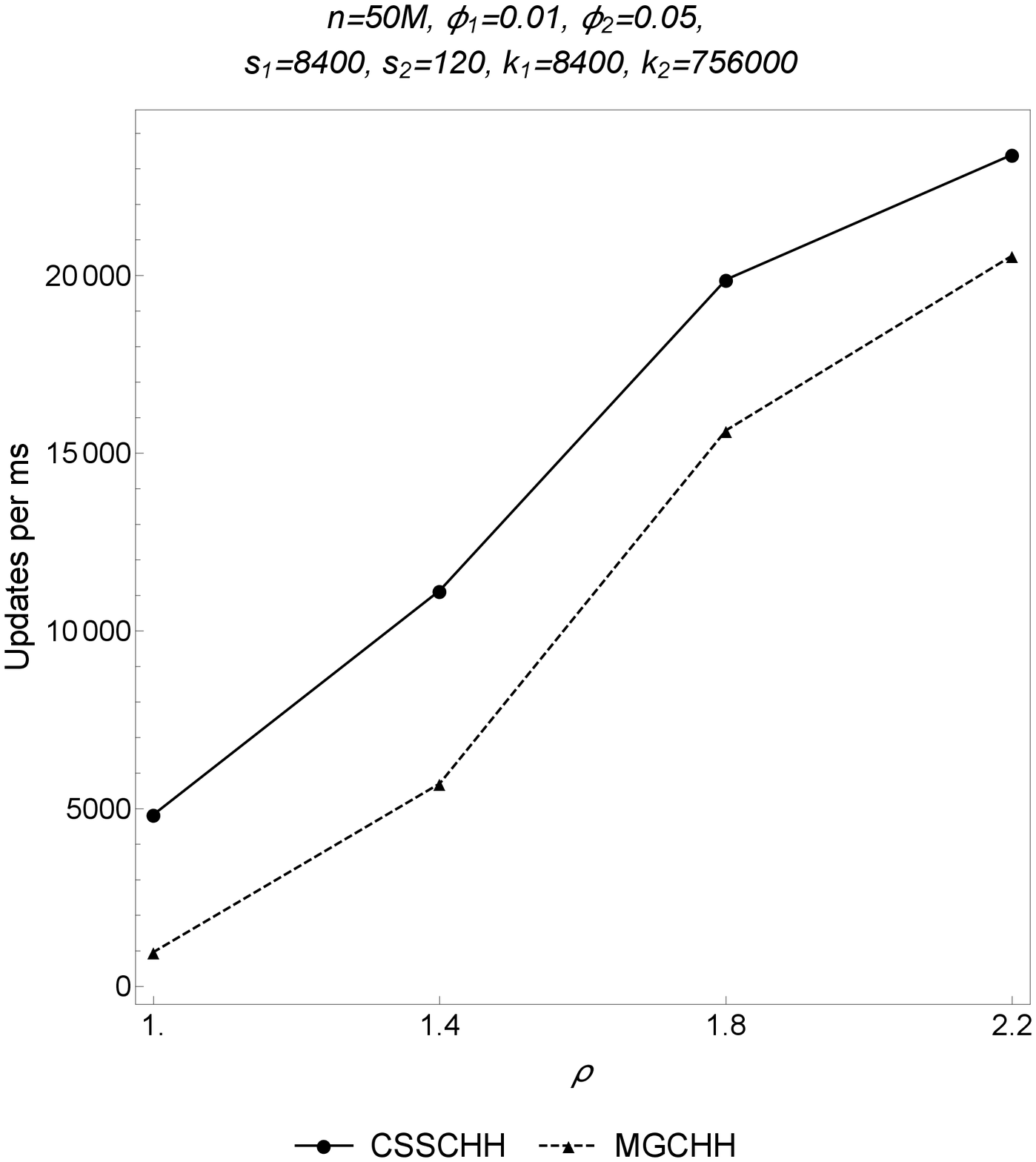}
           \label{sk-update}
        } &
        
      \subfloat[varying the space used ]{
           \includegraphics[width=0.3\textwidth]{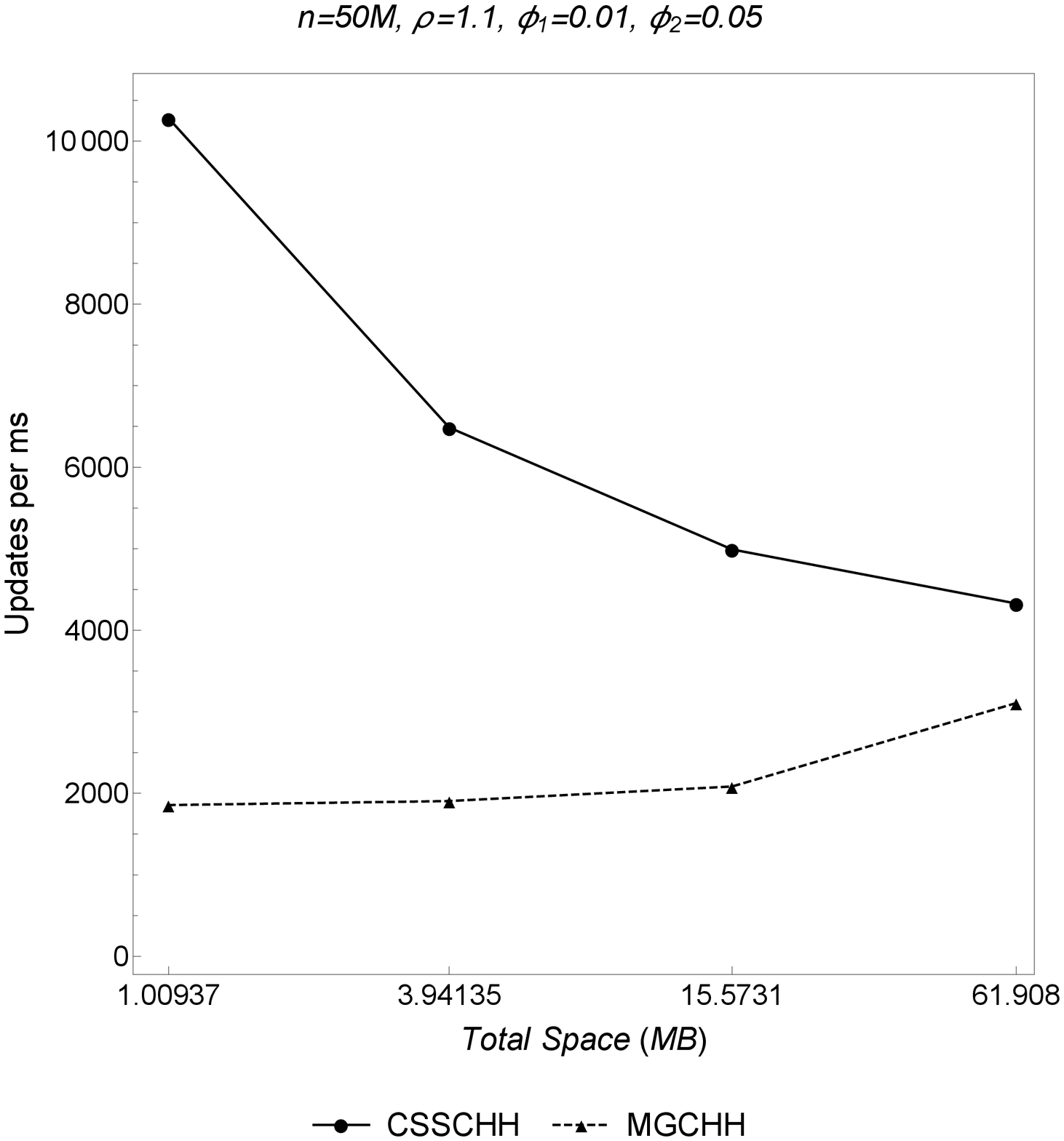}
           \label{size-update}
        } \\

        \subfloat[$\phi_1 = 0.1$]{
           \includegraphics[width=0.3\textwidth]{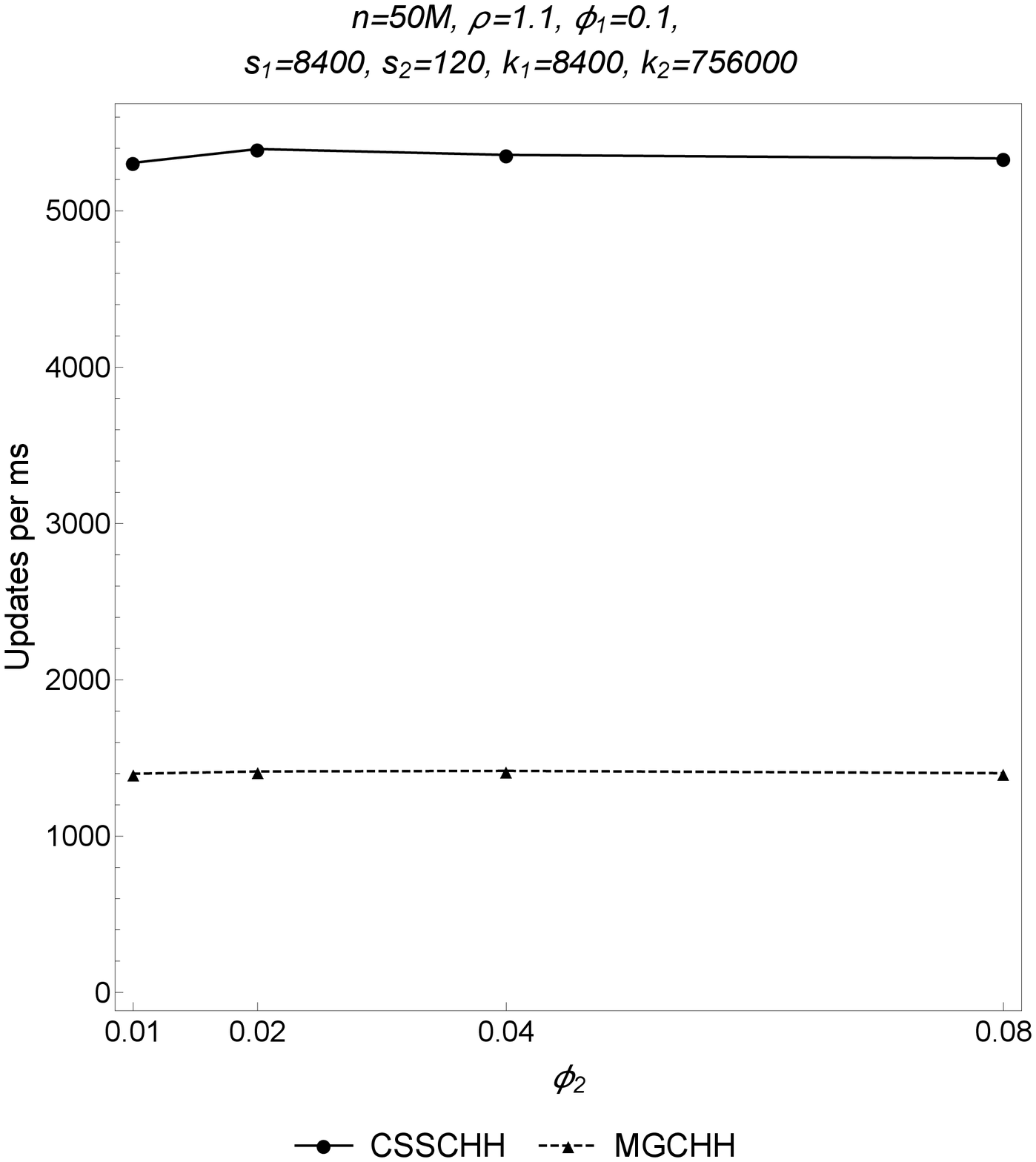}
           \label{phi1-update}
        } &

      \subfloat[$\phi_1 = 0.01$]{
           \includegraphics[width=0.3\textwidth]{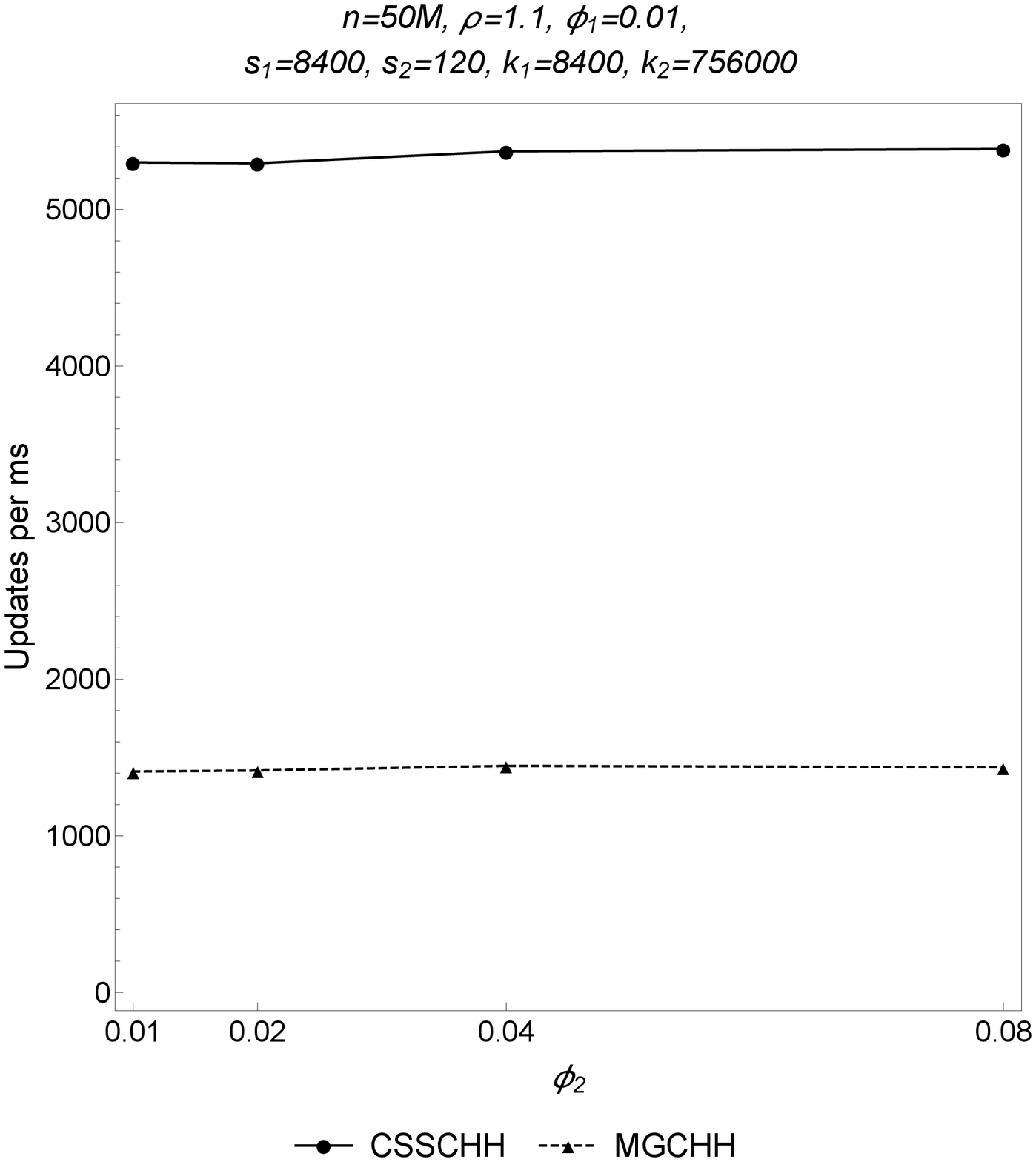}
           \label{phi2-update}
        } &
        
      \subfloat[$\phi_1 = 0.001$]{
           \includegraphics[width=0.3\textwidth]{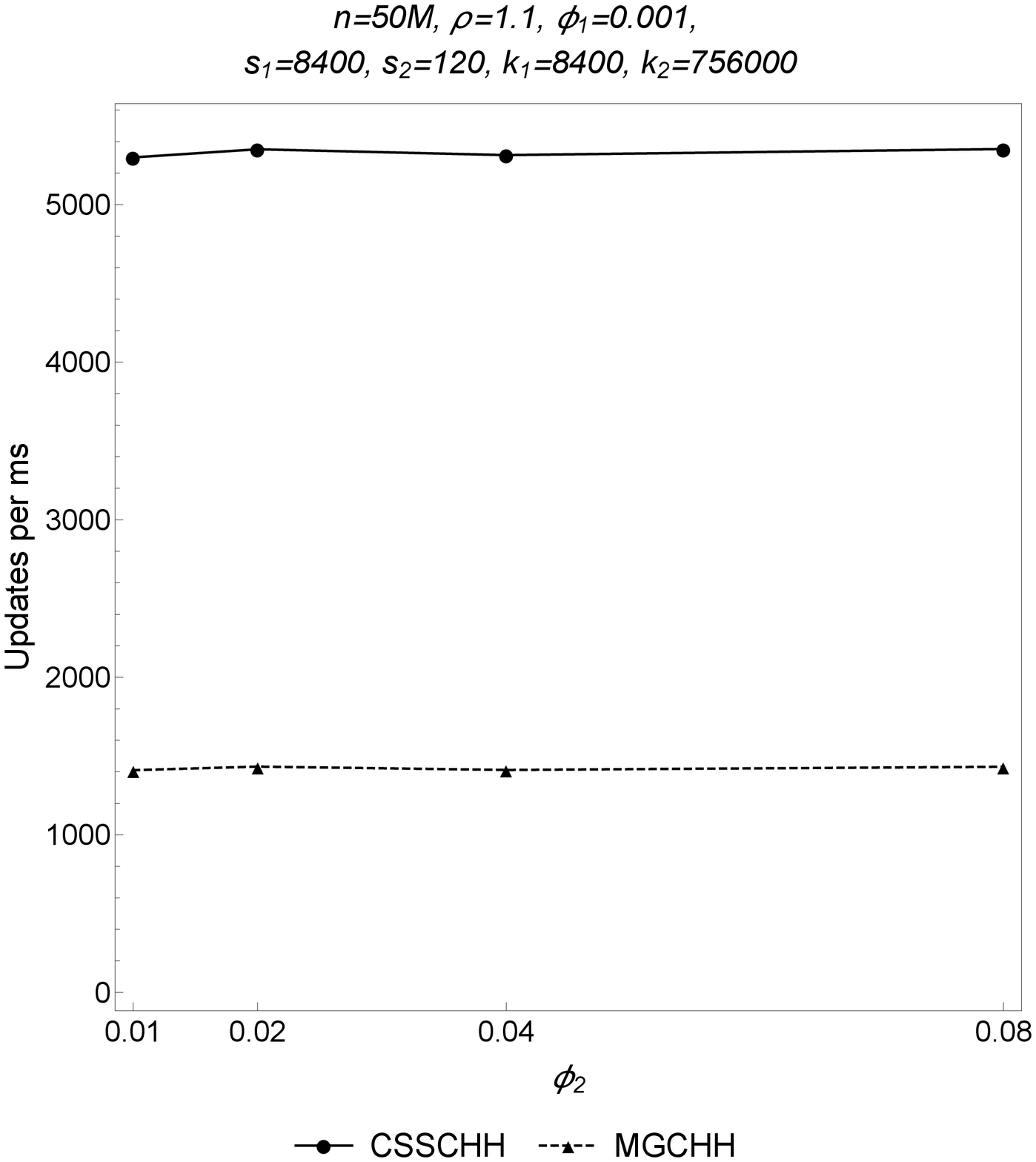}
           \label{phi3-update}
        } 
\end{tabular}
 
 \caption{Updates/ms (mean and confidence interval)} 
 \label{Updates}
\end{figure}

\begin{figure}[]
  \centering
  \begin{tabular}{cc}
  
  \subfloat[Precision]{
           \includegraphics[width=0.45\textwidth]{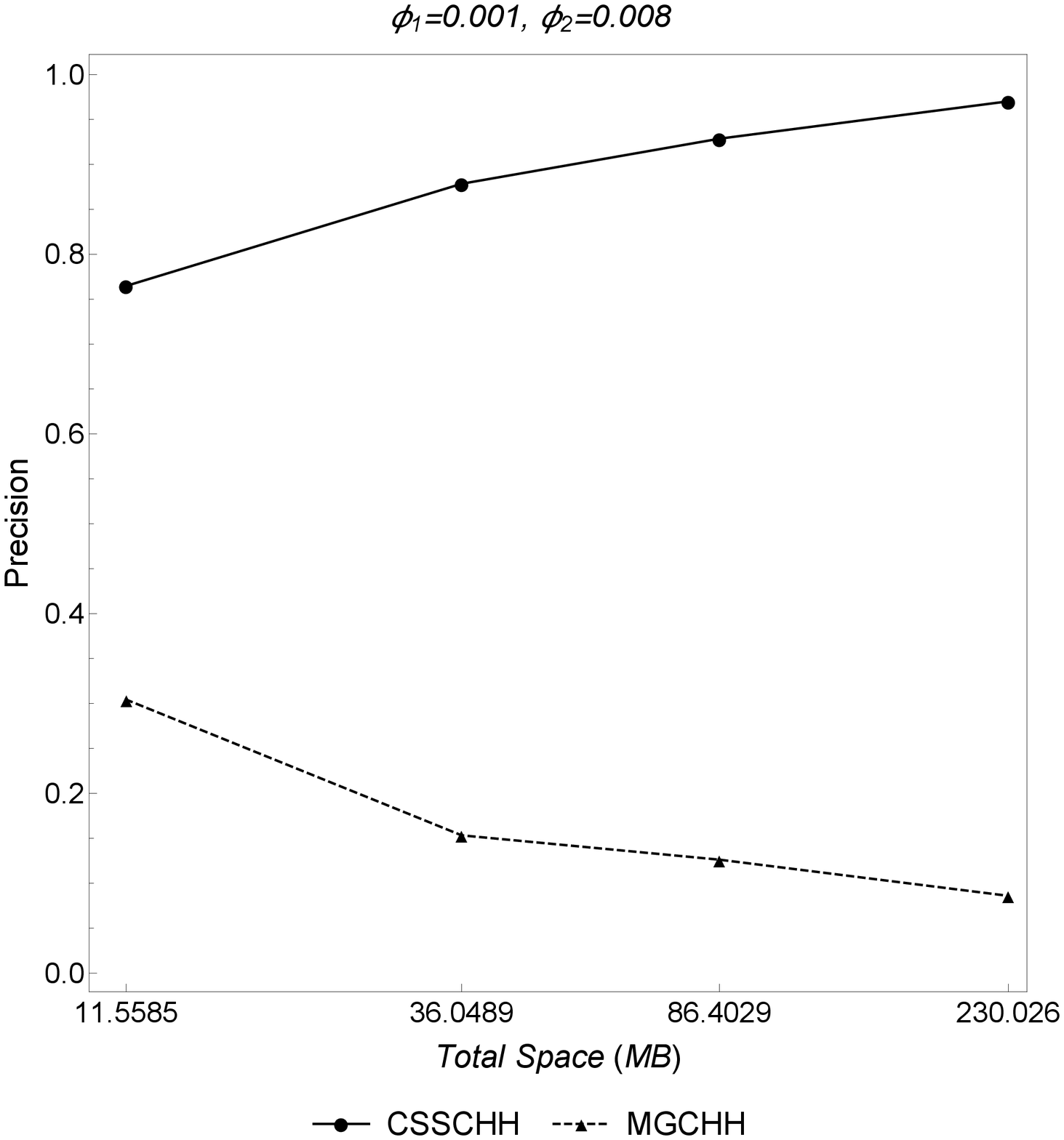}
           \label{wc-space-prec}
        } &

        \subfloat[Updates/ms]{
           \includegraphics[width=0.45\textwidth]{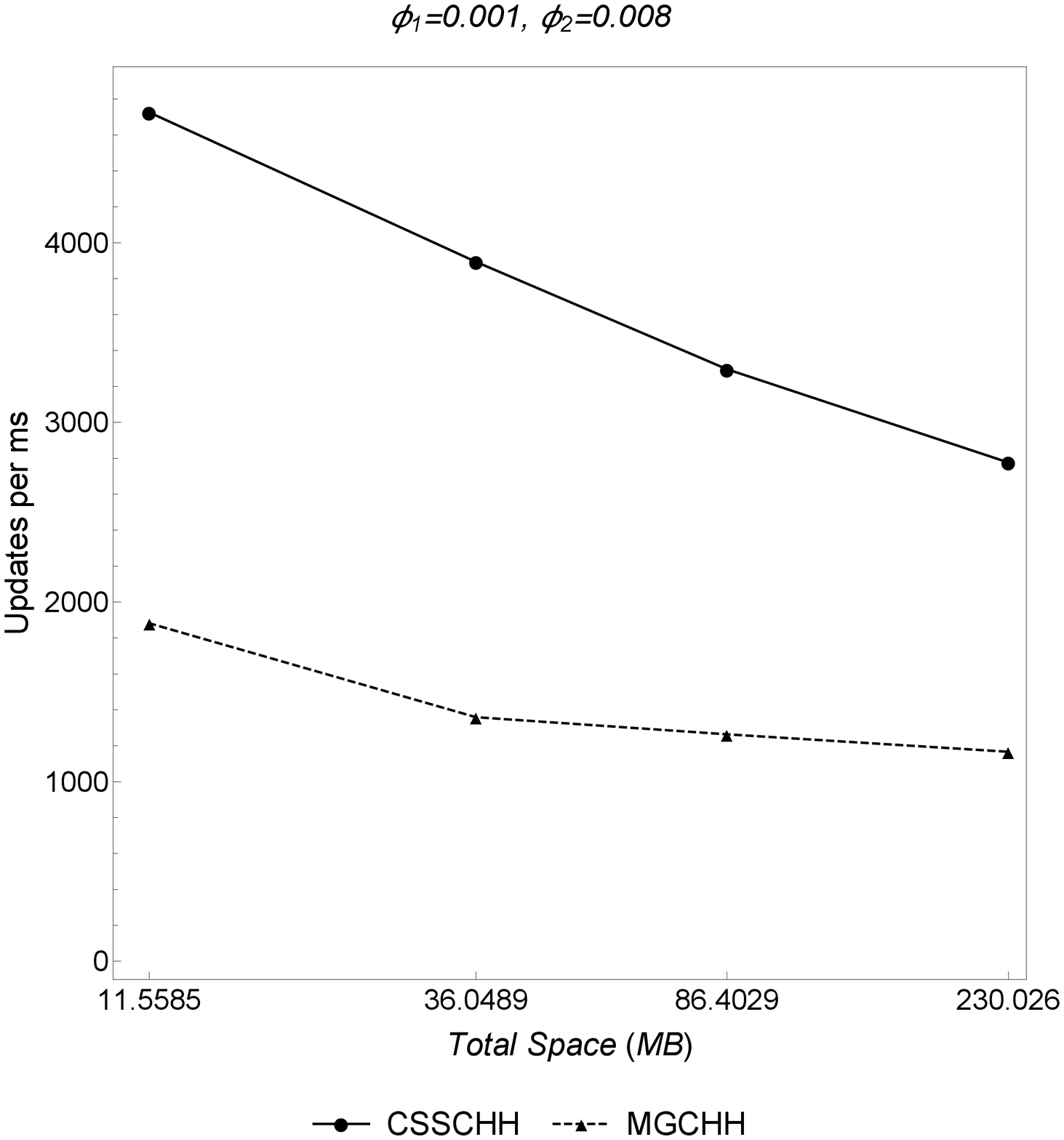}
           \label{wc-space-updates}
        } \\
        
      \subfloat[Relative Error]{
           \includegraphics[width=0.45\textwidth]{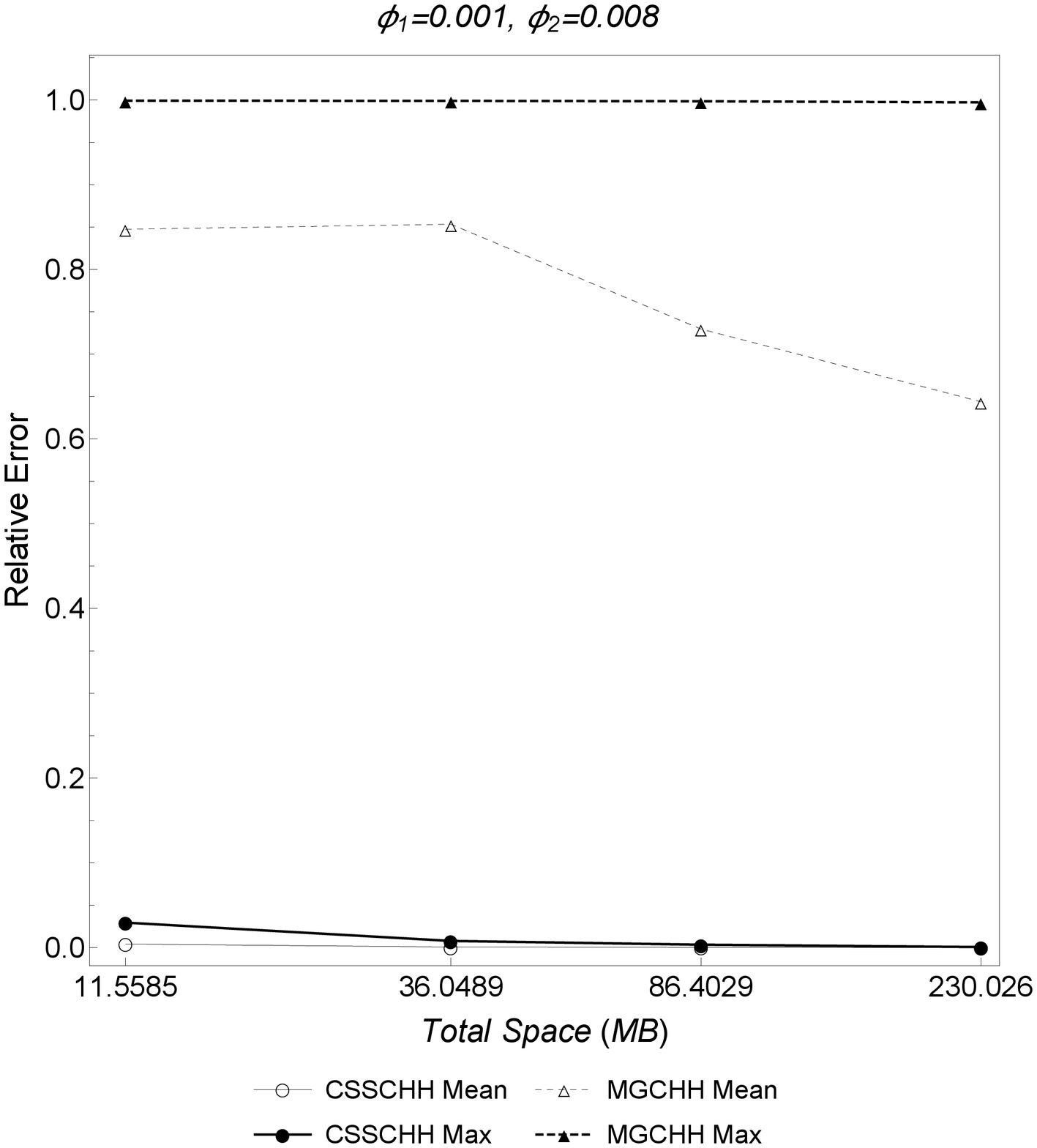}
           \label{wc-space-relerr}
        } &

      \subfloat[Absolute Error]{
           \includegraphics[width=0.45\textwidth]{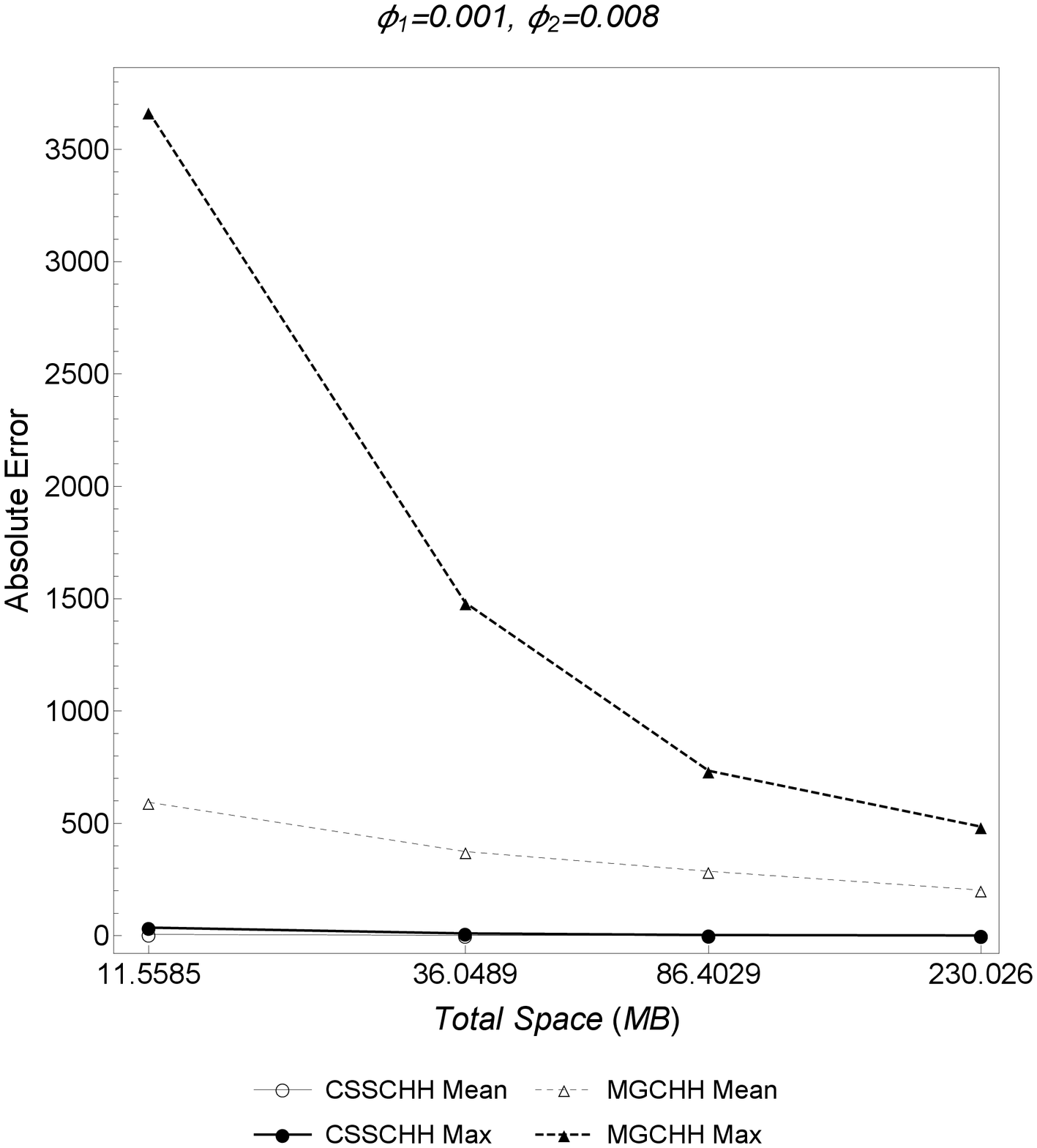}
           \label{wc-space-abserr}
        } \\

      \end{tabular}
 
 \caption{Worldcup'98 results obtained varying the space allowed (mean and confidence interval)} 
 \label{wc-space}
\end{figure}

\begin{figure}[]
  \centering
  \begin{tabular}{cc}
  
  \subfloat[Precision]{
           \includegraphics[width=0.45\textwidth]{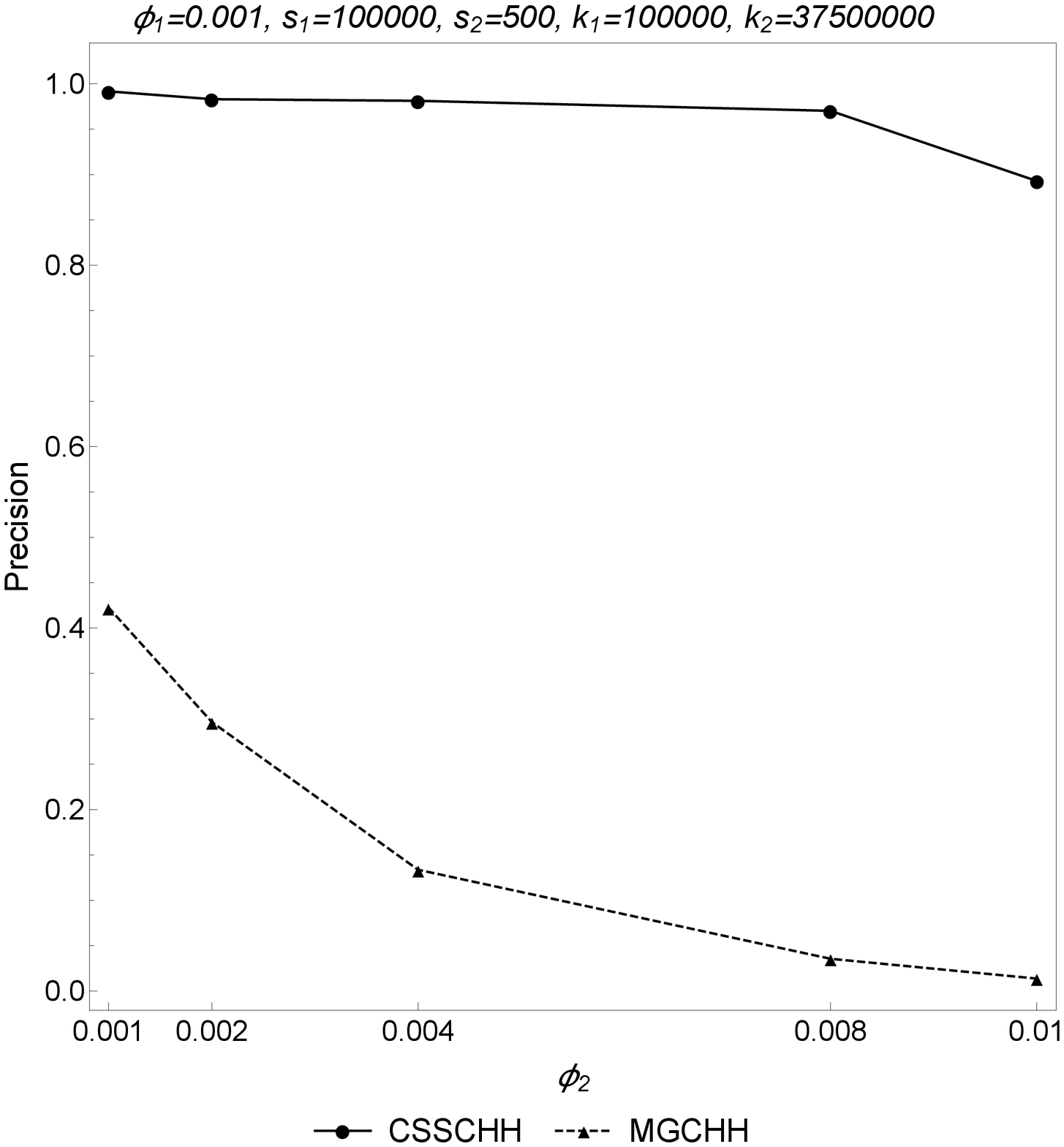}
           \label{wc-phi2-prec}
        } &

        \subfloat[Updates/ms]{
           \includegraphics[width=0.45\textwidth]{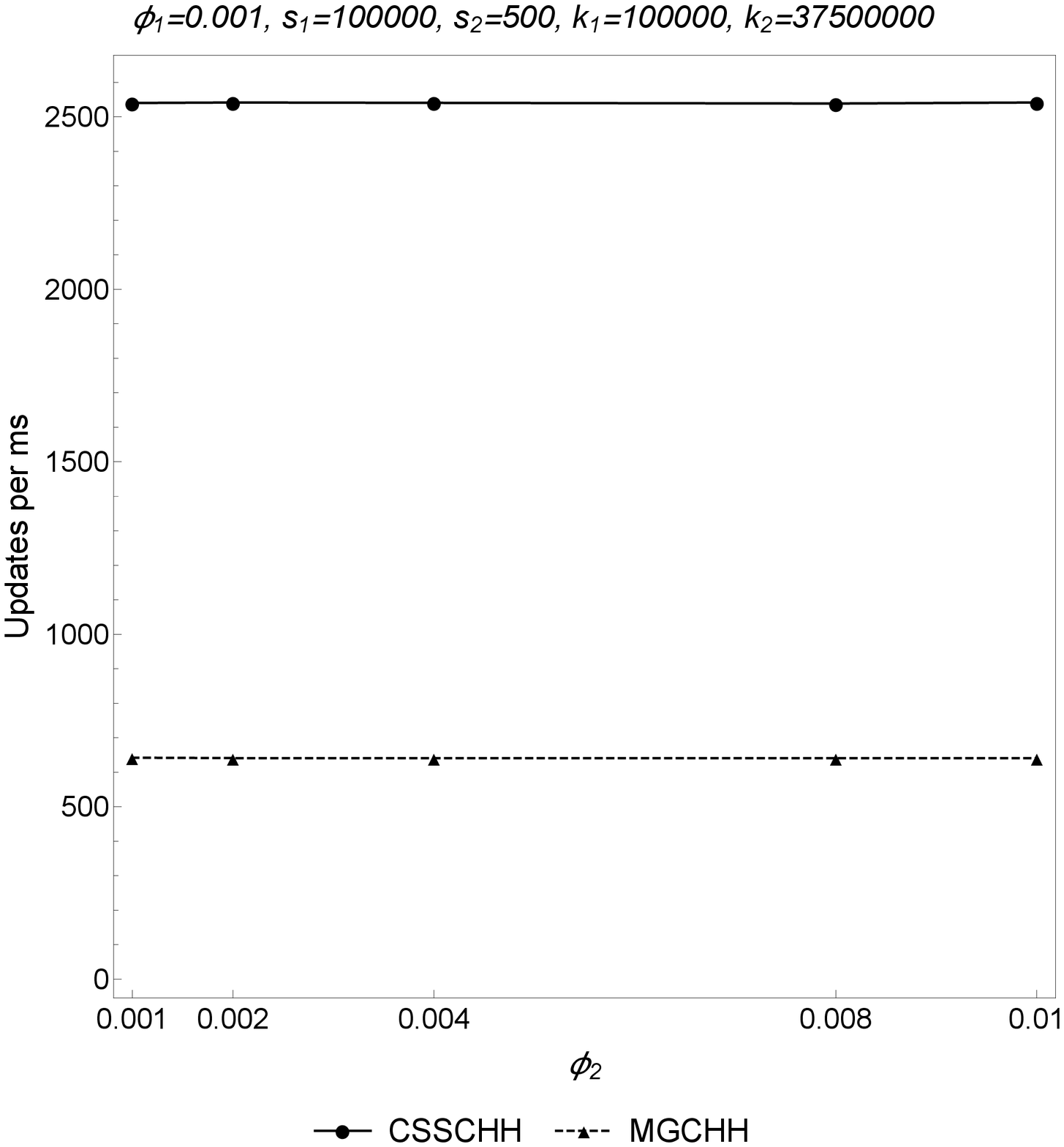}
           \label{wc-phi2-updates}
        } \\

      \subfloat[Relative Error]{
           \includegraphics[width=0.45\textwidth]{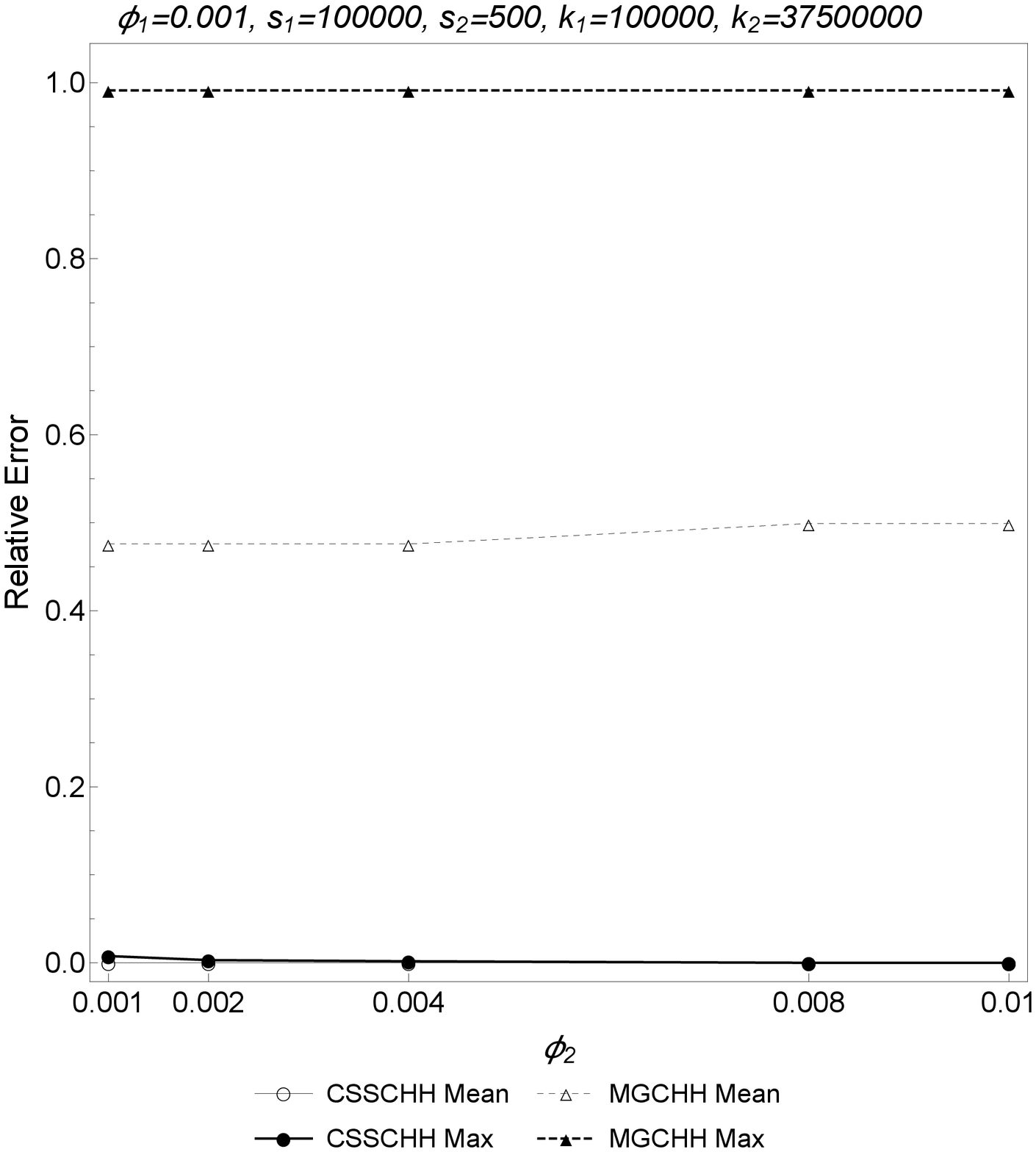}
           \label{wc-phi2-relerr}
        } &

      \subfloat[Absolute Error]{
           \includegraphics[width=0.45\textwidth]{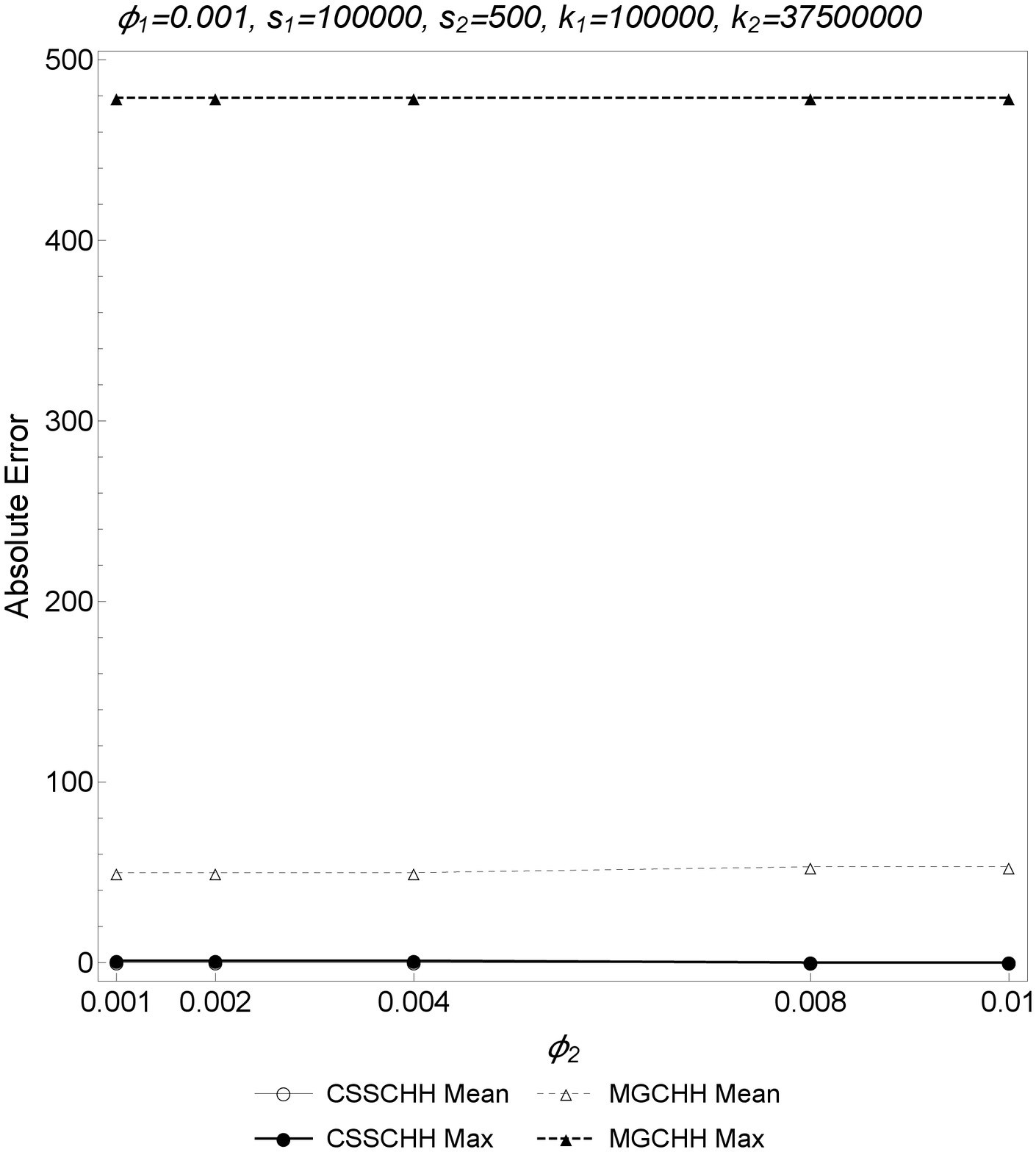}
           \label{wc-phi2-abserr}
        } \\
        
      \end{tabular}
 
 \caption{Worldcup'98 results obtained varying the $\phi_2$ threshold (mean and confidence interval)} 
 \label{wc-phi2}
\end{figure}

\section{Conclusions}
\label{conclusions}

In this paper, we have studied the problem of mining Correlated Heavy Hitters from a two-dimensional data stream. We have presented CSSCHH, a new counter-based algorithm for tracking CHHs, and formally proved its error bounds and correctness. We have compared our algorithm to MGCHH, a recently designed deterministic algorithm based on the Misra--Gries algorithm both from a theoretical point of view and through extensive experimental results, and we have shown that our algorithm outperforms it with regard to accuracy and speed whilst requiring asymptotically much less space. 

\clearpage

\bibliographystyle{plain}
\bibliography{bibliography}

\end{document}